\documentclass[preprint,showpacs,preprintnumbers,amsmath,amssymb]{revtex4}

\usepackage{graphicx}
\usepackage{dcolumn}
\usepackage{bm}

\newtheorem{theorem}{Theorem}[section]
\newtheorem{lemma}[theorem]{Lemma}

\newenvironment{proof}[1][Proof]{\begin{trivlist}
\item[\hskip \labelsep {\bfseries #1}]}{\end{trivlist}}

\newcommand{\qed}{\nobreak \ifvmode \relax \else
      \ifdim\lastskip<1.5em \hskip-\lastskip
      \hskip1.5em plus0em minus0.5em \fi \nobreak
      \vrule height0.75em width0.5em depth0.25em\fi}

\begin{document}

\preprint{}

\title{Moment-Based Evidence for Simple Rational-Valued Hilbert-Schmidt Generic $2 \times 2$ Separability Probabilities}
\author{Paul B. Slater (with appendix by Charles F. Dunkl)}%
\email{slater@kitp.ucsb.edu}
\affiliation{%
University of California, Santa Barbara, CA 93106-4030\\
}%
\date{\today}

\begin{abstract}
Employing Hilbert-Schmidt measure, we explicitly compute and analyze 
a number of determinantal product (bivariate) moments $|\rho|^k |\rho^{PT}|^n$, $k,n=0,1,2,3,\ldots$, $PT$ denoting partial transpose, for both generic (9-dimensional) two-rebit ($\alpha=\frac{1}{2}$) and generic 
(15-dimensional) two-qubit  ($\alpha=1$)  density matrices 
$\rho$. The results are, then, incorporated by Dunkl into  a  {\it general} formula 
(Appendix~\ref{ComplexConjectures}), parameterized by $k, n$ and 
$\alpha$, with the case $\alpha=2$,  presumptively corresponding to generic (27-dimensional) quaternionic systems. Holding the Dyson-index-like parameter 
$\alpha$ fixed, the induced {\it univariate} moments 
 $(|\rho| |\rho^{PT}|)^n$ and $|\rho^{PT}|^n$ are inputted into a Legendre-polynomial-based (least-squares) probability-distribution reconstruction algorithm of Provost ({\it Mathematica J.}, {\bf{9}}, 727 (2005)), yielding $\alpha$-specific separability probability estimates. Since, as the number of inputted moments grows, estimates based on  the variable $|\rho| |\rho^{PT}|$ strongly decrease,  while ones employing 
$|\rho^{PT}|$ strongly increase (and converge faster),  the gaps between upper and lower estimates diminish, yielding sharper and sharper bounds. Remarkably, 
for $\alpha = 2$, with the use of 2,325 moments, a separability-probability lower-bound 0.999999987 as large as $\frac{26}{323} \approx 0.0804954$ is found. For $\alpha=1$, based on 2,415 moments, a lower bound results that is 0.999997066 times as large as 
$\frac{8}{33} \approx 0.242424É$, a (simpler still) fractional value that had previously been conjectured ({\it J. Phys. A}, {{\bf  40}}, 14279 (2007)). Furthermore, for $\alpha=\frac{1}{2}$, employing 3,310 moments, the lower bound is 0.999955 times as large as  $\frac{29}{64} = 0.453125$, a rational value previously considered  ({\it J. Phys. A}, {{\bf  43}}, 195302 (2010)).
\end{abstract}


\maketitle

\tableofcontents
\section{Introduction}
 In a much cited paper \cite{ZHSL}, {\. Z}yczkowski, Horodecki,
Sanpera and Lewenstein expanded upon 
``three main reasons''--``philosophical'', ``practical''
and ``physical''--for attempting to evaluate the probability that mixed states of composite quantum systems are separable in nature. Pursuing such a research agenda, it was conjectured \cite[sec. IX]{slater833}--based on "a confluence of
numerical and theoretical results"--that the separability probabilities of generic (15-dimensional) two-qubit 
and (9-dimensional) two-rebit quantum systems, in terms of the Hilbert-Schmidt/Euclidean/flat (HS) measures \cite{szHS,ingemarkarol}, are $\frac{8}{33} \approx 0.242424$ and $\frac{8}{17} \approx 0.470588$, respectively.
In this study, we shall avail ourselves of newly-proposed formulas of Dunkl 
(Appendix~\ref{appDunkl}) for (bivariate) moments 
of products of determinants of density matrices ($\rho$) and of their partial transposes ($\rho^{PT}$) \cite{asher,michal} to investigate these hypotheses from a novel  perspective, as well as extend our analyses beyond the strictly two-rebit and two-qubit frameworks.
(To be fully explicit, we note here that both [symmetric] two-rebit and [Hermitian] two-qubit $4 \times 4$ density matrices $\rho$ have unit trace and nonnegative eigenvalues, while their partial transposes $\rho^{PT}$ can be obtained by transposing in place the four $2 \times 2$ blocks of $\rho$. The Hilbert-Schmidt metric--from which the corresponding measure can, of course, be derived--is defined by the line
element squared, $\frac{1}{2} \mbox{Tr}[(\mbox{d}\rho)^2]$ \cite[eq. (14.29)]{ingemarkarol}.)

Reconstructions of probability distributions based on these product 
moment formulas of Dunkl do prove to be  highly supportive of the specific  HS two-qubit conjecture 
(sec.~\ref{TWOQUBITS}), while definitively ruling out its two-rebit counterpart (sec.~\ref{TWOREBITS}), but emphatically not a later advanced value of 
$\frac{29}{64} = 0.453125$ \cite[p. 6]{advances}.
Extending these analyses from the real ($\alpha=\frac{1}{2}$) and complex ($\alpha=1$) cases to the (presumptively, since we lack relevant computer-algebraic determinantal moment calculations) generic (27-dimensional) quaternionic ($\alpha=2$) instance \cite{asher2,adler,baez,batle2}, in which the off-diagonal entries of the $4 \times 4$ density matrices can be quaternions, we find that the value 
$\frac{26}{323} \approx 0.0804954$ fits our moment-based computations, may we say, amazingly well (sec.~\ref{TWOQUATBITS}). Nevertheless, the apparently formidable challenges of rigorously proving  the determinantal moment formulas of Dunkl and/or the conjectured simple fractional separability probabilities certainly remain.
(To again be explicit, the only {\it rigorously} demonstrated results reported in this paper are those we have been able to obtain through computer algebraic [Mathematica]
methods--using the Cholesky-decomposition parameterization of $\rho$--for the moments of $|\rho|^k |\rho^{PT}|^n$ for $n=1, 2, \ldots,13$ for the two-rebit systems and $n=1, 2, 3, 4$ for the two-qubit systems [sec.~\ref{DMDPM}], and $n=1$ for their qubit-qutrit [$6 \times 6$] counterparts [sec.~\ref{QubitQutrit}], as well as $n=1,\ldots,10$ for minimally degenerate two-rebit systems [sec.~\ref{minimally}]. Aside from the presentation and discussion of these results, the paper is concerned with the [unproven] generalization to arbitrary $n$ by Dunkl of these specific results, and its apparent successful application in  probability-distribution reconstruction procedures [sec.~\ref{PROVOSTreconstruction}]. This latter step is taken in order to examine anew and extend certain conjectures as to the specific values of the separability probabilities, the properties of which were first investigated by {\. Z}yczkowski, Horodecki, Sanpera and Lewenstein \cite{ZHSL}.)

In marked contrast to the finite-dimensional focus in this study on $2 \times 2$ quantum systems (and, marginally, on $2 \times 3$ systems [sec.~\ref{QubitQutrit}]), let us note the (asymptotically-based) 
conclusion of Ye that "the probability of finding separable quantum states within quantum states is extremely small and the Peres-Horodecki PPT criterion as tools to detect separability is imprecise for large $N$, in the sense of both Hilbert-Schmidt and Bures volumes" \cite[p. 14]{ye2}. (The Bures distance measures the length of a curve within the cone of positive operators on the Hilbert space 
\cite[sec. 9.4]{ingemarkarol}, while the Bures volume of the set of mixed states is remarkably equal to the volume of an $(N^2-1)$-dimensional hypersphere of radius $\frac{1}{2}$ \cite[p. 351]{ingemarkarol}.)
Also, contrastingly, to the predominantly "nondegenerate/full-rank" objectives here (cf. sec.~\ref{minimally}), Ruskai and Werner have demonstrated that 'bipartite states of low rank are almost surely entangled" 
\cite{RuskaiWerner}.
\section{Density-Matrix Determinantal Product Moments} \label{DMDPM}
Let us begin our investigation into the indicated statistical aspects of the "geometry of quantum states" \cite{ingemarkarol,BookReview} by noting the two following special cases--which will be extended in certain bivariate directions--of the (univariate determinantal moment) formulas \cite{csz}[eq. (3.2)] 
(cf. \cite[Theorem 4]{andai}):
\begin{equation} \label{firstold}
\left\langle |\rho|^k \right\rangle_{2-rebit/HS}=945 \Big( 4^{3-2 k} 
\frac{ \Gamma (2 k+2) \Gamma (2 k+4)}{\Gamma (4 k+10)} \Big)
\end{equation}
and 
\begin{equation} \label{secondold}
\left\langle |\rho|^k \right\rangle_{2-qubit/HS}=108972864000 \frac{ \Gamma (k+1) \Gamma (k+2) \Gamma (k+3) \Gamma
   (k+4)}{\Gamma (4 (k+4))},
\end{equation}
$k=0, 1, 2,\dots$
The bracket notation $\left\langle \right\rangle$ is employed to denote expected value, 
while $\rho$ indicates a generic
(symmetric) two-rebit or generic (Hermitian) two-qubit ($4 \times 4$) density matrix. The expectation is taken with respect to the probability distribution determined by the 
Hilbert-Schmidt/Euclidean/flat metric on either the 9-dimensional space of generic two-rebit or 15-dimensional space of generic two-qubit systems \cite{szHS,ingemarkarol}.

At the outset of our study, we were able to compute {\it seventeen} (thirteen two-rebit and four two-qubit) non-trivial (bivariate) extensions of these two formulas, involving now in addition to $|\rho|$, the quantum-theoretically important determinant $|\rho^{PT}|$. (The nonnegativity 
of $|\rho^{PT}|$--as a corollary of  the celebrated Peres-Horodeccy results \cite{asher,michal}--constitutes a necessary and sufficient condition for separability/disentanglement, when $\rho$ is a $4 \times 4$ density matrix \cite{augusiak,azumaban}.) At this point of our presentation, we note that three of these seventeen extensions are expressible--incorporating as the last factors on their right-hand sides, the two formulas above ((\ref{firstold}), (\ref{secondold}))--as 
\begin{equation} \label{firstnew}
\left\langle |\rho|^k |\rho^{PT}| \right\rangle_{2-rebit/HS}= 
\frac{(k-1) (k (2 k+11)+16)}{32 (k+3) (4 k+11) (4 k+13)} \left\langle |\rho|^k \right\rangle_{2-rebit/HS},
\end{equation}
\begin{equation} \label{secondnew}
\left\langle |\rho|^k |\rho^{PT}|^2 \right\rangle_{2-rebit/HS}= 
\frac{k (k (k (k (4 k (k+12)+203)+368)+709)+2940)+4860}{1024 (k+3) (k+4)
   (4 k+11) (4 k+13) (4 k+15) (4 k+17)} \left\langle |\rho|^k \right\rangle_{2-rebit/HS}
\end{equation}
and
\begin{equation} \label{thirdnew}
\left\langle |\rho|^k |\rho^{PT}| \right\rangle_{2-qubit/HS}= 
\frac{k (k (k+6)-1)-42}{8 (2 k+9) (4 k+17) (4 k+19)} \left\langle |\rho|^k \right\rangle_{2-qubit/HS}.
\end{equation}
These three new formulas were, initially, established by "brute force" computation--that is  calculating the first ($k=0, 1, 2,\ldots,15$ or so) instances of them, then employing
the Mathematica command FindSequenceFunction, and verifying the formulas
generated on still higher values of $k$. 

Let us note here the ranges of the two variables of central interest, $|\rho| \in [0,\frac{1}{256}]$ and 
$|\rho^{PT}| \in [-\frac{1}{16},\frac{1}{256}]$. For various analytical and conventional purposes, it is often convenient to have variables defined over the unit interval [0,1]. If we so (linearly) transform the two determinantal variables, then the rational factors on the right-hand sides of  (\ref{firstnew}) and (\ref{secondnew}) get replaced, respectively, by
\begin{equation}
\frac{8 (k (k (34 k+297)+867)+842)}{17 (k+3) (4 k+11) (4 k+13)}
\end{equation}
and
\begin{equation}
\frac{64 (k (k (k (k (68 k (17 k+348)+200835)+904492)+2279781)+3048904)+1689900)}{289
   (k+3) (k+4) (4 k+11) (4 k+13) (4 k+15) (4 k+17)}.
\end{equation}
\section{The mixed/balanced 
variable $|\rho| |\rho^{PT}|= |\rho \rho^{PT}|$}
As a special case ($k=1$) of formula (\ref{firstnew}), we obtain the rather remarkable moment result, zero, already reported in \cite{HSorthogonal}. The immediate interpretation of this finding is that for the generic two-rebit systems, the two determinants
$|\rho|$ and $|\rho^{PT}|$ comprise a pair of nine-dimensional {\it orthogonal} polynomials \cite{dunkl2,dumitriu,griffithsspano} with respect to Hilbert-Schmidt measure. (C. Dunkl has kindly pointed out that orthogonality here does not imply zero {\it correlation}. The analogous
quantity for generic two-{\it qubit} systems is {\it not} zero, however, but 
$-\frac{1}{4576264}$.) 
In addition to this first ($k=1$) 
HS zero-moment of the  product variable
$|\rho| |\rho^{PT}|$ in the two-rebit case, we had been able to compute its higher-order  moments, $k =2,\ldots,6$. (The result for $k=2$, that is $\frac{7}{5696343244800}$, can be obtained by direct application of formula (\ref{secondnew}).) 
\subsection{Range of Variable}
The feasible range of the (mixed/balanced) variable is 
$|\rho| |\rho^{PT}| \in [-\frac{1}{110592},\frac{1}{256^2}]$--the lower bound of which$-\frac{1}{110592} =- 2^{-12} 3^{-3}$.
This lower bound, determined by analyzing a general convex combination of a Bell state and the fully-mixed state, can be achieved with the entangled two-rebit density matrix
\begin{equation} \label{NewMatrix}
\rho= 
\left(
\begin{array}{cccc}
 \frac{1}{6} & -\frac{1}{6 \sqrt{2}} & \frac{1}{6
   \sqrt{2}} & \frac{1}{12} \left(-1+\sqrt{3}\right) \\
 -\frac{1}{6 \sqrt{2}} & \frac{1}{3} & \frac{1}{12}
   \left(-1-\sqrt{3}\right) & -\frac{1}{6 \sqrt{2}} \\
 \frac{1}{6 \sqrt{2}} & \frac{1}{12}
   \left(-1-\sqrt{3}\right) & \frac{1}{3} & \frac{1}{6
   \sqrt{2}} \\
 \frac{1}{12} \left(-1+\sqrt{3}\right) & -\frac{1}{6
   \sqrt{2}} & \frac{1}{6 \sqrt{2}} & \frac{1}{6}
\end{array}
\right).
\end{equation}
The determinant of $\rho$  here is 
$\frac{1}{576} \left(2 \sqrt{3}-3\right) \approx 0.000805732$
and that of its partial transpose,
$\frac{1}{576} \left(-3-2 \sqrt{3}\right) \approx -0.0112224$ 
(their product being $-\frac{1}{110592} \approx -9.04225 \cdot 10^{-6}$).
Both $\rho$  and $\rho^{PT}$ here have three identical eigenvalues 
($\frac{1}{12} \left(3-\sqrt{3}\right) \approx 0.105662$ for $\rho$ and 
$\frac{1}{12} \left(3+\sqrt{3}\right) \approx 0.394338$ for $\rho^{PT}$).
The isolated eigenvalues for $\rho$  and $\rho^{PT}$ are $\frac{1}{4} \left(1+\sqrt{3}\right) 
\approx 0.683013$, and  $\frac{1}{4} \left(1-\sqrt{3}\right) \approx -0.183013$, respectively. The purity (index of coincidence \cite[p. 56]{ingemarkarol}) of (\ref{NewMatrix}) equals $\frac{1}{2}$, so the participation 
ratio is 2. Its concurrence 
 is $\frac{1}{2} \left(\sqrt{3}-1\right) \approx 0.366025$, while its entanglement of formation is 
\cite[sec. 15.7]{ingemarkarol}
\begin{equation}
E_{complex}[\rho]=\frac{\log (2) \left(\log \left(84+48
   \sqrt{3}\right)-\sqrt{3} \log (3)\right)}{4 \log
   \left(1+\frac{1}{\sqrt{3}}\right) \log   \left(1+\sqrt{3}\right)} \approx 1.21665.
\end{equation}
(Supportively, Dunkl has noted that the computed zeros in his Gaussian quadrature analyses (sec.~\ref{Gaussian} of the two-rebit case fit well into the known ranges of 
$|\rho|$ and $|\rho^{PT}|$.) Alternatively, taking into account the {\it real} nature of the entries of $\rho$, in the sense of the "foil" theory of Caves, Fuchs and Rungta 
\cite{carl}, one 
has a 
concurrence of \cite[eq. (4)]{batleplastino1} $-\frac{1}{\sqrt{3}}$, and an entanglement of formation 
\cite[eq. (2)]{batleplastino1} of
\begin{equation}
E_{real}[\rho]=\frac{\log (2) \left(\log (1728)+2 \sqrt{6} \tanh
   ^{-1}\left(\sqrt{\frac{2}{3}}\right)\right)}{6 \log
   \left(6-2 \sqrt{6}\right) \log \left(2
   \left(3+\sqrt{6}\right)\right)} \approx 6.56825.
\end{equation}
\section{Contour plots of bivariate probability distributions}
For the further edification of the reader, we present in Fig.~\ref{fig:JointRebit}  a numerically-generated contour plot of the joint Hilbert-Schmidt 
(bivariate) probability distribution of $|\rho|$ and $|\rho^{PT}|$ in the two-rebit 
case, and in  Fig.~\ref{fig:JointQubit}, its two-qubit analogue. (A colorized grayscale output is employed, in which larger values appear lighter.) In 
Fig.~\ref{fig:JointDiff} is displayed the {\it difference} obtained by subtracting the second (two-qubit) distribution from the first (two-rebit) 
distribution.
\begin{figure}
\includegraphics{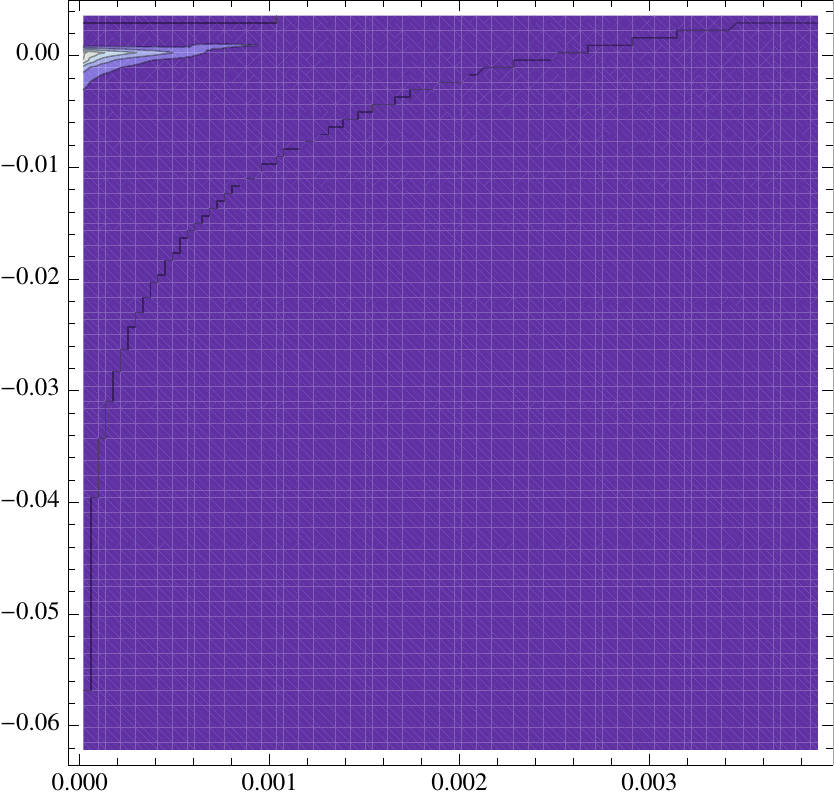}
\caption{\label{fig:JointRebit}Contour plot of the joint Hilbert-Schmidt probability distribution of $|\rho|$ (horizontal axis) and $|\rho^{PT}|$ in the two-rebit case. Larger values appear lighter. The variable ranges are $|\rho| \in [0,\frac{1}{256}]$ and $|\rho^{PT}| \in [-\frac{1}{16},\frac{1}{256}]$. One billion random density matrices were employed.}
\end{figure}
\begin{figure}
\includegraphics{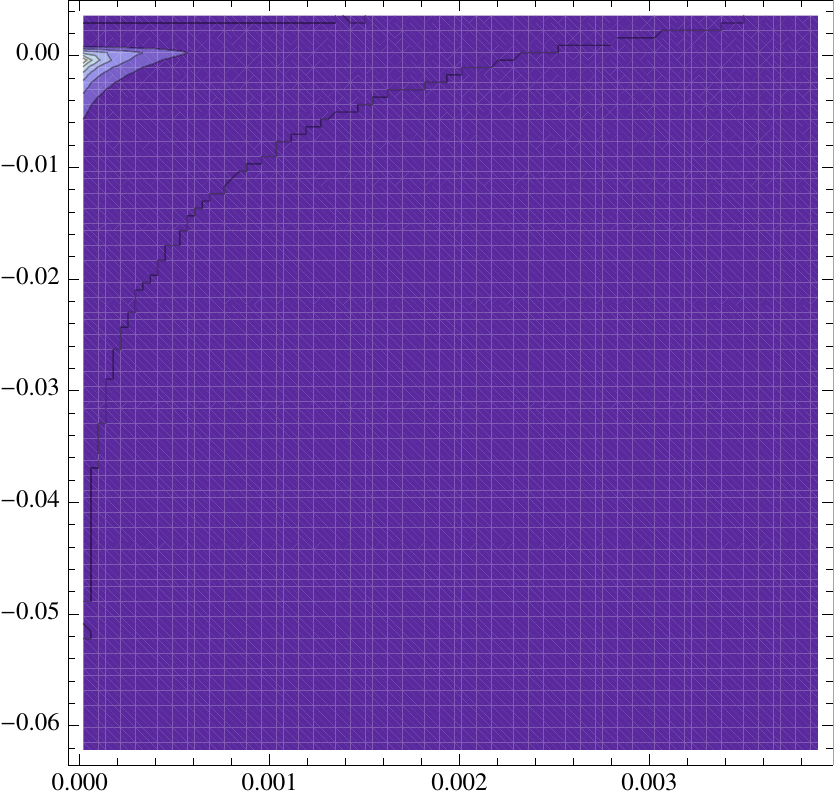}
\caption{\label{fig:JointQubit}Contour plot of the joint Hilbert-Schmidt probability distribution of $|\rho|$ (horizontal axis) and $|\rho^{PT}|$ in the two-qubit case. Six hundred million random density matrices were employed.}
\end{figure}
\begin{figure}
\includegraphics{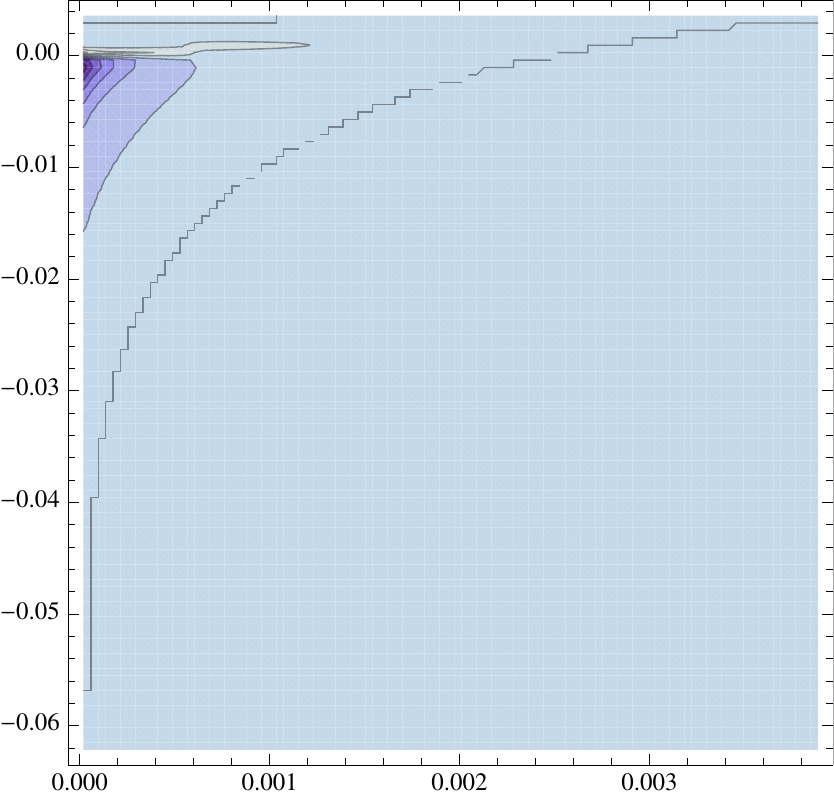}
\caption{\label{fig:JointDiff}Difference obtained by subtracting 
the two-qubit HS probability distribution in Fig.~\ref{fig:JointQubit} from the two-rebit probability distribution in Fig.~\ref{fig:JointRebit}. Darker colors indicate more negative values.}
\end{figure}
(The black curves in all three contour plots appear to be attempts by Mathematica to establish the nonzero-zero probability boundaries--which, it would, of course, be of interest to explicitly determine/parameterize, if possible--of the joint domain of $|\rho|$ and $|\rho^{PT}|$.)

These  last three figures are based on  Hibert-Schmidt sampling (utilizing Ginibre ensembles \cite{csz}) of random density matrices, using $10,000=100^2$ bins. In regard to the two-qubit plot, 
K. {\.Z}yzckowski informally wrote:
"A high peak in the upper corner means that: a) a majority of the entangled states is 'little entangled' (small $det(\rho^T)$) or rather, they are 'close' to the boundary of the set, so
one eigenvalue is close to zero, and the determinant is small; 
b) as $det(\rho)$ is also small, it means that these entangled states
live  close to the boundary of the set of all states (at least one eigenvalue is very small), 
but this is very much consistent with the observation that the center of the convex body of the 2-qubit states is separable  (so entangled states have to live  'close' to the boundary).
Similar reasoning has to hold in the real case as well."

\section{Determinantal product moment formulas}
\subsection{Two-rebit case}
At a still later point in our investigation, we realized that we might make further
progress--despite apparent limitations on the number of determinantal moments we could explicitly compute--by exploiting the evident pattern followed by our newly-found formulas 
(\ref{firstnew}) and (\ref{secondnew})--in particular, the structure in 
their
denominators. This encouragingly proved to be the case, as we were able
to additionally establish that 
\begin{equation} \label{fourthnew}
\left\langle |\rho|^k |\rho^{PT}|^3 \right\rangle_{2-rebit/HS}= 
\frac{A_3}{B_3} \left\langle |\rho|^k \right\rangle_{2-rebit/HS},
\end{equation}
where
\begin{equation}
A_3=8 k^9+180 k^8+1674 k^7+8559 k^6+29493 k^5+84291 k^4+136801 k^3-401334
   k^2-2516616 k-3612816
\end{equation}
and 
\begin{equation} 
B_3=32768 (k+3) (k+4) (k+5) (4 k+11) (4 k+13) (4 k+15) (4 k+17) (4 k+19) (4
   k+21).
\end{equation}

So, it then became rather evident that we can write for general non-negative integer $n$,
\begin{equation} \label{generalRebit}
\left\langle |\rho|^k |\rho^{PT}|^n \right\rangle_{2-rebit/HS}= 
\frac{A_n}{B_n} \left\langle |\rho|^k \right\rangle_{2-rebit/HS},
\end{equation}
where both the numerator $A_n$ and the denominator $B_n$ are $3 n$-degree polynomials (thus, forming a "biproper rational function" \cite{chou}) in $k$ (the leading coefficient of $A_n$ being $2^n$), and
\begin{equation} \label{denominator}
B_n=128^{n } (k+3)_{n } \left(2 k+\frac{11}{2}\right)_{2 n },
\end{equation}
where the Pochhammer symbol $(x)_n \equiv \frac{\Gamma(x+n)}{\Gamma(x)} = x (x+1)\ldots(x+n-1)$ is employed.
Further still, moving upward to the next level ($n=4$), we determined that
\begin{equation} \label{fifthnew}
\left\langle |\rho|^k |\rho^{PT}|^4 \right\rangle_{2-rebit/HS}= 
\frac{A_4}{B_4} \left\langle |\rho|^k \right\rangle_{2-rebit/HS},
\end{equation}
where
\begin{equation}
A_4= 16 k^{12}+576 k^{11}+9112 k^{10}+84496 k^9+525681 k^8+2389416 k^7+7805462
   k^6+13904508 k^5+
\end{equation}
\begin{displaymath}
+6212189 k^4+166748972 k^3+1636873812 k^2+5496485760
   k+6610161600,
\end{displaymath}
and $B_4$ is given by (\ref{denominator})  with $n=4$. 
The real part of one of the roots of  $A_4$ is 2.999905, suggesting to us some possible interesting asymptotic behavior of the roots of these numerators, $n \rightarrow \infty$. In a related 
predecessor study \cite{HSorthogonal}[sec. II.B.2], we had been able to discern the general structure that the {\it denominators} of  certain "intermediate [rational] functions" used in computing the (univariate) moments of $\left\langle \rho^{PT}|^n \right\rangle_{2-rebit/HS}$, $n=1,\ldots,9$ followed.

From our four new two-rebit determinantal moment results (\ref{firstnew}), (\ref{secondnew}), (\ref{fourthnew}) and (\ref{fifthnew}),  we see that the {\it constant} terms in the $3 n$-degree numerator $A_n$
are $-16, 4860, -3612816$ and $6610161600$ for $n=1, 2, 3, 4$. Since we had previously computed \cite{HSorthogonal}[eqs, (33)-(41)]
the moments of $\left\langle |\rho^{PT}|^n \right\rangle_{2-rebit/HS}$, 
$n=1,\ldots,9$, we were also immediately able to determine the next five members
of this sequence $\{-16, 4860, -3612816, 6610161600 \}$. However, no general rule for this sequence, which would, interestingly, directly allow us to obtain a formula for $\left\langle |\rho^{PT}|^n \right\rangle_{2-rebit/HS}$, had yet emerged for them.

Certainly, it would be of interest to conduct analyses parallel to those reported above for metrics of quantum-information-theoretic interest other than the Hilbert-Schmidt, such as the Bures (minimal monotone) metric 
\cite{szBures,ingemarkarol,andai}. The computational challenges involved, however, might, at least in certain respects, be even more substantial.
\subsection{Use of Cholesky decomposition in rigorously finding 
formulas for general $k$}
After having posted the results above, along with additional ones, as a preprint \cite{slaterJoint}, Charles Dunkl detailed a computational proposal that 
he had outlined to us somewhat earlier. The particularly attractive feature of this proposal was that it would--holding the exponent $n$ of 
$|\rho^{PT}|$ fixed--be able to compute the adjustment factors for {\it general} $k$, rather than having to do so for sufficient numbers of {\it individual} members of the sequence $k=1,\ldots, N$, so that we could successfully apply the Mathematica command FindSequenceFunction, as had been our strategy heretofore. The proposal of Dunkl (Appendix~\ref{appDunkl}) involved parameterizing $4 \times 4$ density matrices in terms of their Cholesky decompositions. The parameters (ten in number for the two-rebit case and sixteen for the two-qubit case) would be viewed as points on the surface  of a  unit (due to the trace requirement)  
10-sphere or 16-sphere. The squares of the points lie in a simplex.
One can then employ the corresponding Dirichlet probability distributions over the simplices to determine the associated expected values (joint moments).
(A further highly facilitating aspect here is that both $|\rho|$ and the jacobian for the transformation to Cholesky variables are simply {\it monomials} in the variables.)
Using this approach, we were able to extend our single ($n=1, \alpha=1$) two-{\it qubit} result (\ref{thirdnew}) to the $n=2$ case,
\begin{equation} \label{secondnewqubit}
\left\langle |\rho|^k |\rho^{PT}|^2  \right\rangle_{2-qubit/HS}= 
\end{equation}
\begin{displaymath}
\frac{k (k (k (k (k (k+15)+67)+45)+220)+4260)+10944}{64 (2 k+9) (2 k+11)
   (4 k+17) (4 k+19) (4 k+21) (4 k+23)} \left\langle |\rho|^k \right\rangle_{2-qubit/HS}.
\end{displaymath}

Additionally, in the following array, 
\begin{equation} \label{CoefficientArray1}
\left(
\begin{array}{cccccc}
 -16 & 4860 & -3612816 & 6610161600 & -23680812672000 &
   147885533254368000 \\
 5 & 2940 & -2516616 & 5496485760 & -21644930613600 & 144374531813568000
   \\
 9 & 709 & -401334 & 1636873812 & -7755993054000 & 58524043784903280 \\
 2 & 368 & 136801 & 166748972 & -1199508017652 & 11977854861441312 \\
 - & 203 & 84291 & 6212189 & -4378482660 & 1052189083196640 \\
 - & 48 & 29493 & 13904508 & 29246867605 & -30302414250528 \\
 - & 4 & 8559 & 7805462 & 7876634465 & -6899036908859 \\
 - & - & 1674 & 2389416 & 2649513956 & 3583820785224 \\
 - & - & 180 & 525681 & 883461210 & 1632448582425 \\
 - & - & 8 & 84496 & 219916945 & 477741210624 \\
 - & - & - & 9112 & 40679505 & 118164517947 \\
 - & - & - & 576 & 5660714 & 23817008856 \\
 - & - & - & 16 & 575800 & 3786901675 \\
 - & - & - & - & 40000 & 469728096 \\
 - & - & - & - & 1680 & 44685468 \\
 - & - & - & - & 32 & 3143808 \\
 - & - & - & - & - & 153360 \\
 - & - & - & - & - & 4608 \\
 - & - & - & - & - & 64
\end{array}
\right)
\end{equation}
we show ($n=1,\ldots,6$), column-by-column, the 
$(3 n +1)$ coefficients of the numerator polynomials in ascending order--the entries in the first row corresponding to the constant terms,\ldots--in the two-rebit case. 

Additional results for the cases $n=7,\ldots,13$ were found \cite[eqs. (17)-(21)]{slaterJoint}.
The leading (highest-order) coefficients in these thirteen sets of two-rebit results were found to be expressible in descending order as 
\begin{equation} \label{specfunct1}
C_{3 n+1}=2^n; \hspace{.1in} C_{3 n}=3 \times 2^{n-1} n (n+2);\hspace{.1in} C_{3 n-1}=2^{n -3} n  (n  (n  (9 n +32)+24)-45);
\end{equation}
\begin{equation}
 C_{3 n-2}=
2^{n -4} n  \left(n  \left(n  \left(n  \left(9 n ^2+42
   n +52\right)-119\right)-52\right)-60\right).
\end{equation}

From these four formulas, we are able to reconstruct ($n=1$) all four entries in the first column of the table (\ref{CoefficientArray1}). Thus, it appears that, in general, $C_{3 n-i}$ is a polynomial in $n$ of degree
$2 (i+1)$. (For $i=3 n -1$, we obtain the constant term, of strong interest. With the full knowledge of all the constant terms, and none of the other 
coefficients, we could obtain the univariate moments 
$\left\langle |\rho^{PT}|^n \right\rangle_{2-rebit/HS}$.)
Further, we have found that 
\begin{equation}
 C_{3 n-3}= 
\end{equation}
\begin{displaymath}
\frac{1}{5} 2^{n -7} (n -1) \left(135 n ^7+855 n
   ^6+1895 n ^5-1771 n ^4-3091 n ^3-7731 n ^2+32394
   n \right),
\end{displaymath}
and 
\begin{equation} \label{specfunct2}
 C_{3 n-4}= \frac{1}{5} 2^{n -8} (n -1) n
\end{equation}
\begin{displaymath} 
n  (n  (n  (n  (n  (3 n  (3 n  (9 n
   +59)+377)-2887)-2295)-10535)+112240)-181492)+436720.
\end{displaymath}
\subsection{Two-qubit formulas}
The numerators of our four sets ($n =1, 2, 3, 4$) of two-qubit results (the first two having been obtained by "brute force" Mathematica computations, and the last two, using the Cholesky-decomposition parameterization) are expressible, in similar fashion, as 
\begin{equation} 
\left(
\begin{array}{cccc}
 -42 & 10944 & -6929280 & 9247219200 \\
 -1 & 4260 & -3684384 & 6039653760 \\
 6 & 220 & -456948 & 1342859616 \\
 1 & 45 & 80168 & 64072440 \\
 - & 67 & 27783 & -13235252 \\
 - & 15 & 5373 & 1080858 \\
 - & 1 & 1458 & 1160375 \\
 - & - & 282 & 278478 \\
 - & - & 27 & 50991 \\
 - & - & 1 & 7542 \\
 - & - & - & 749 \\
 - & - & - & 42 \\
 - & - & - & 1
\end{array}
\right) .
\end{equation}
We observe that the leading coefficients $C_{3 n +1}$ of all four numerators are 1, so they are {\it monic} in character,  
while the next-to-leading coefficients  fit the pattern 
$C_{3 n}  = 3 n (n+3)/2$.

It is evident at this point, in striking analogy to the general two-rebit formula (\ref{generalRebit}), that in the two-qubit scenario,
\begin{equation} \label{denominator2}
\left\langle |\rho|^k |\rho^{PT}|^n \right\rangle_{2-qubit/HS}= 
\frac{\hat{A}_n}{\hat{B}_n} \left\langle |\rho|^k \right\rangle_{2-qubit/HS},
\end{equation}
where, again, both the numerator $\hat{A}_n$ and the denominator $\hat{B}_n$ are $3 n$-degree polynomials in $k$, and 
(cf. (\ref{denominator}))
\begin{equation}
\hat{B}_n= 2^{6 n } \left(k+\frac{9}{2}\right)_{n } \left(2
   k+\frac{17}{2}\right)_{2 n }.
\end{equation}

\section{Determinantal Product  Moment formulas for $6 \times 6$ density matrices} \label{QubitQutrit}
Of course, one may also consider issues analogous to those discussed above
for bipartite quantum systems of higher dimensionality. To begin such a course of analysis, we have found 
for the generic real $6 \times 6$ ("rebit-retrit") density matrices (occupying a 20-dimensional space) the result
\begin{equation} \label{rebitretrit}
\left\langle |\rho|^k |\rho^{PT}| \right\rangle_{rebit-retrit/HS}= 
\frac{4 k^5+40 k^4+95 k^3-220 k^2-1149 k-1170}{576 (k+4) (3 k+11) (3
   k+13) (6 k+23) (6 k+25)}
 \left\langle |\rho|^k \right\rangle_{rebit-retrit/HS}.
\end{equation}
Increasing the  exponential parameter $n$ from 1 to 2, we obtained that the rational function adjustment factor for $\left\langle |\rho|^k |\rho^{PT}|^2 \right\rangle_{rebit-retrit/HS}$ is the ratio of
\begin{equation}
16 k^9+336 k^8+2616 k^7+8496 k^6+12069 k^5+101979 k^4+903539 k^3+3316809
   k^2+5620320 k+3715740
\end{equation}
to another ninth-degree polynomial 
\begin{equation}
331776 (k+5) (3 k+11) (3 k+13) (3 k+14) (3 k+16) (6 k+23) (6 k+25) (6
   k+29) (6 k+31).
\end{equation}

Additionally, for the generic complex $6 \times 6$ (qubit-qutrit) density matrices (occupying a 35-dimensional space), we have obtained the result
\begin{equation} \label{qubitqutrit}
\left\langle |\rho|^k |\rho^{PT}| \right\rangle_{qubit-qutrit/HS}= 
\frac{k^5+15 k^4+37 k^3-423 k^2-2558 k-3840}{72 (2 k+13) (3 k+19) (3
   k+20) (6 k+37) (6 k+41)}
 \left\langle |\rho|^k \right\rangle_{qubit-qutrit/HS}.
\end{equation}

It should be pointed out, however, that in contrast to the $4 \times 4$ 
density matrix case, the nonnegativity of the determinant of the corresponding partial transpose of a $6 \times 6$ density matrix does not guarantee separability, since possibly two eigenvalues of the partial transpose could be negative, indicative of entanglement, while still yielding a nonnegative determinant (cf. \cite{augusiak}).
\section{Minimally degenerate two-rebit density matrices} \label{minimally}
For the {\it eight}-dimensional manifold composed of generic minimally degenerate two-rebit systems (corresponding to density matrices $\rho$ with at least one eigenvalue zero), forming the boundary of the nine-dimensional manifold of generic two-rebit systems, we have computed the Hilbert-Schmidt moments of $|\rho^{PT}|^n, n=1,\ldots,10$. (For such systems, 
$|\rho^{PT}| \in [-\frac{1}{16},\frac{1}{432}]$.) These results are given in Appendix~\ref{AREA}. (Charles Dunkl was able to find rational functions of $k$ for $n=1, 2, 3$--but not yet further--which yielded these 
moments when $k$ was set to zero.)

We note that as a particular case of results of Szarek, Bengtsson  and {\.Z}yczkowski \cite{sbz}, the Hilbert-Schmidt probability that a generic two-rebit system is separable is {\it twice} the HS probability that a generic minimally degenerate two-rebit system is separable. 

\section{Estimation of separability probabilities, using conjectured formulas} \label{PROVOSTreconstruction}
\subsection{Two-rebit case ($\alpha=\frac{1}{2}$)} \label{TWOREBITS}
We now utilize the conjectured formulas 
(App.~\ref{ComplexConjectures})--developed by Dunkl at an intermediate stage in our research effort--with the Dyson-index-type parameter $\alpha$ set to 
$\frac{1}{2}$, corresponding to the two-rebit case. In Fig.~\ref{fig:ListPlotTwoMoments}, we display the 
corresponding Hilbert-Schmidt separability probability estimates obtained by application of the Legendre-polynomial-based probability density reconstruction (Mathematica) procedure of Provost \cite[eq. (15)]{Provost}--yielding least-squares approximating polynomials--to the  sequence of the first 3,310 moments of $(|\rho| |\rho^{PT}|)^n$ (upper blue curve) and to the sequence of the first 3,310 moments
of $|\rho^{PT}|^n$ (lower red curve). (All our computations here and below were conducted with 48-digit accuracy. A uniform "baseline density" was, in effect, assumed, while the use in this capacity of a beta distribution, fitted to the first two moments, and Jacobi polynomials yielded highly erratic estimates when the corresponding Mathematica algorithm of Provost \cite[pp. 750-752]{Provost} was applied.) 
\begin{figure}
\includegraphics{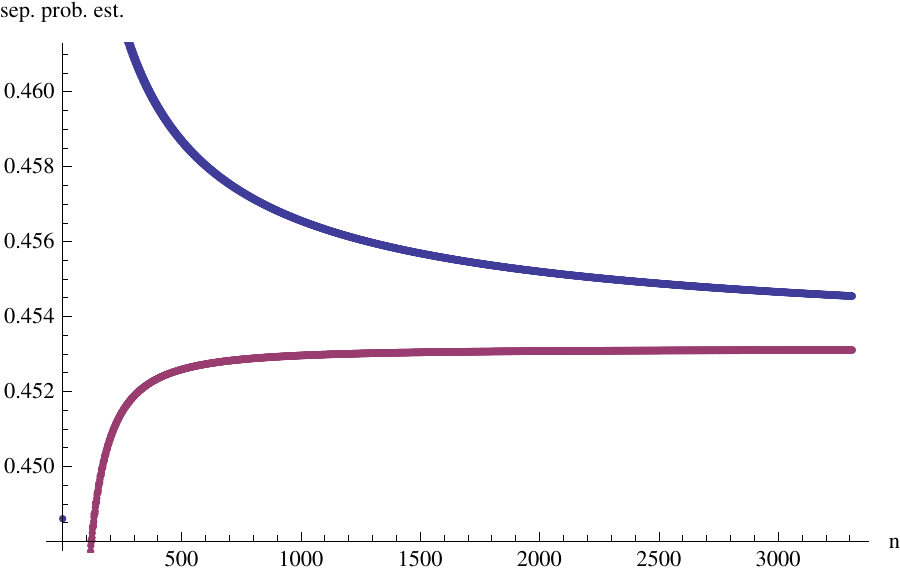}
\caption{\label{fig:ListPlotTwoMoments}Two sets of estimates of the Hilbert-Schmidt two-rebit separability probability. The upper (blue) decreasing curve is based on the first 3,310 (nonnegative) moments of $(|\rho| |\rho^{PT}|)^n$ and the lower 
(red) increasing curve on the first 3,310 (alternating in sign) 
moments of $|\rho^{PT}|^n$. The true separability probability, thus, appears constrained to lie within the range [0.453104500, 0.454543513].}
\end{figure}

In Fig.~\ref{fig:ListPlotTwoMoments}, the last/highest pair of estimates is 
$\{0.453104500, 0.454543513\}$,  so it certainly  appears that the true (common) separability probability for the two variables must lie within this interval. The convergence properties of the two sequences of estimates 
display parallel (increasing-decreasing) behavior in the two-qubit case.  (In sec.~\ref{Gaussian}, Dunkl develops a distinct/alternative probability distribution reconstruction 
approach of interest--which he applies to 
considerably fewer moments than the 3,310 we do--to the two-rebit separability probability estimation problem.)

Our 2007 hypothesis (\cite[sec. X.A]{slater833}) that the Hilbert-Schmidt separability probability of generic two-rebit systems is $\frac{8}{17} \approx 0.470588$ can, thus, be decisively rejected (Fig.~\ref{fig:ListPlotTwoMoments}), since it clearly lies outside the confining interval.  We will here note that in the later 2010 study \cite{advances}[p. 7],  a  numerical estimate of 0.4528427, substantially different from $\frac{8}{17}$, was reported, and it was additionally observed that in 
\cite[sec. V.A.2]{slater833} the best numerical estimate of the two-rebit separability probability obtained there had been 0.4538838. A possible exact value of $\frac{29}{64} = 0.453125$--which does lie within the confining interval in Fig.~\ref{fig:ListPlotTwoMoments}--was, in fact,  suggested in \cite[p. 6]{advances}. Use of linear algebraic principles, did allow us in \cite{advances} to establish an {\it upper} bound
on the generic two-rebit Hilbert-Schmidt separability probability of 
$\frac{1129}{2100} \approx 0.537619$.

We note, importantly, that the lower bound of the confining interval, 0.4531014500 is 0.999955 times as large as $\frac{29}{64}$.
\subsection{Two-qubit case ($\alpha=1$)} \label{TWOQUBITS}
In Fig.~\ref{fig:ListPlotComplexTwoMoments} we similarly show--for the two-qubit 
case ($\alpha=1$)--the estimates obtained by application of the 
probability distribution reconstruction procedure of Provost \cite[eq. (15)]{Provost}
to sequences of 2,415 moments of $(|\rho| |\rho^{PT}|)^n$ (upper blue curve) and $|\rho^{PT}|^n$ (lower red curve).
We, of course, note that the lower bound obtained of 0.2424235313 seems to nicely support our 2007 hypothesis (\cite[sec. X.B]{slater833}) that the Hilbert-Schmidt separability probability of generic two-qubit systems is $\frac{8}{33} \approx 0.242424É$. (The ratio of this lower bound to that based on 2,414 moments is 1.000000006779, indicative of strong convergence. The analogous ratio for the upper estimate was 0.99999153401--somewhat less strong.)
\begin{figure}
\includegraphics{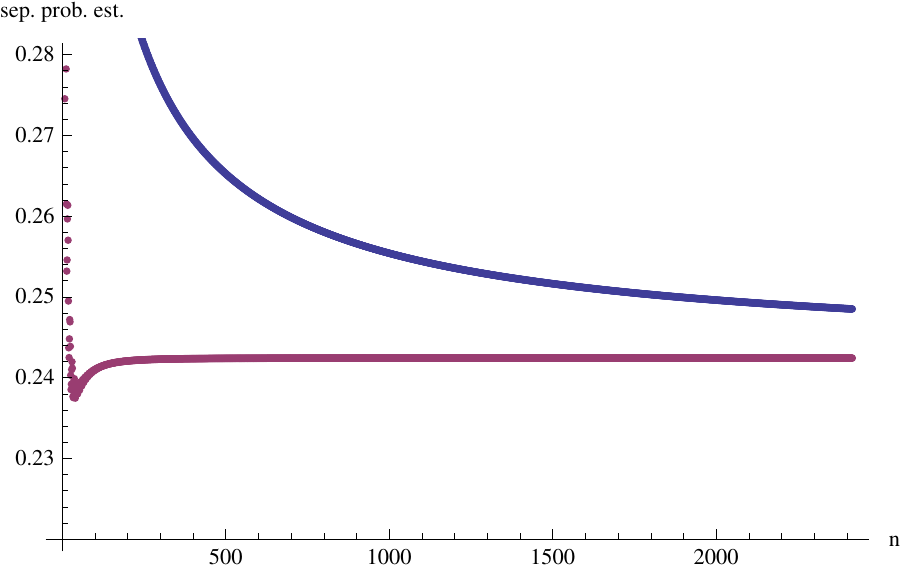}
\caption{\label{fig:ListPlotComplexTwoMoments}Two sets of estimates of the Hilbert-Schmidt two-qubit separability probability based on 2,415 moments. The upper (blue) decreasing curve is based on the (nearly all nonnegative) moments of $(|\rho| |\rho^{PT}|)^n$ and the lower (red) increasing curve on the (alternating in sign) 
moments of $|\rho^{PT}|^n$. The true separability probability, thus,  appears to lie within the confining range [0.2424235313, 0.2485026468].)}
\end{figure}

{\. Z}yczkowski, Horodecki,
Sanpera and Lewenstein, in their foundational paper \cite[eq. (36)]{ZHSL},  provided a numerical estimate--$0.632 \pm 0.002$--of the generic two-qubit separability probability, using  as a measure the product of the uniform distribution on the 3-simplex of eigenvalues and the Haar measure on the 15-dimensional $4 \times 4$ unitary matrices. (The  $4 \times 4$ density matrices were, then, in a sense, over-parameterized. The authors were "surprised" that the probability exceeded $50\%$.) They also advanced \cite[eq. (35)]{ZHSL} certain analytical 
arguments that the probability was in the interval [0.302, 0.863]. While these studies are of great conceptual interest, they did not specifically employ as measures those defined by the volume elements of metrics of interest (such as the Hilbert-Schmidt, Bures,\ldots) over the quantum states.

\subsection{Reconstructed probability distributions}
In Figs.~\ref{fig:MnatPDFRebits} and \ref{fig:MnatPDFQubits}, we show (based on 200 moments, using now the procedure of Mnatsakanov 
\cite{mnatsakanov2}), rather than that of Provost \cite{Provost}, the reconstructed HS two-rebit and two-qubit probability distributions for both sets of moments, all distributions linearly transformed to the interval [0,1].
\begin{figure}
\includegraphics{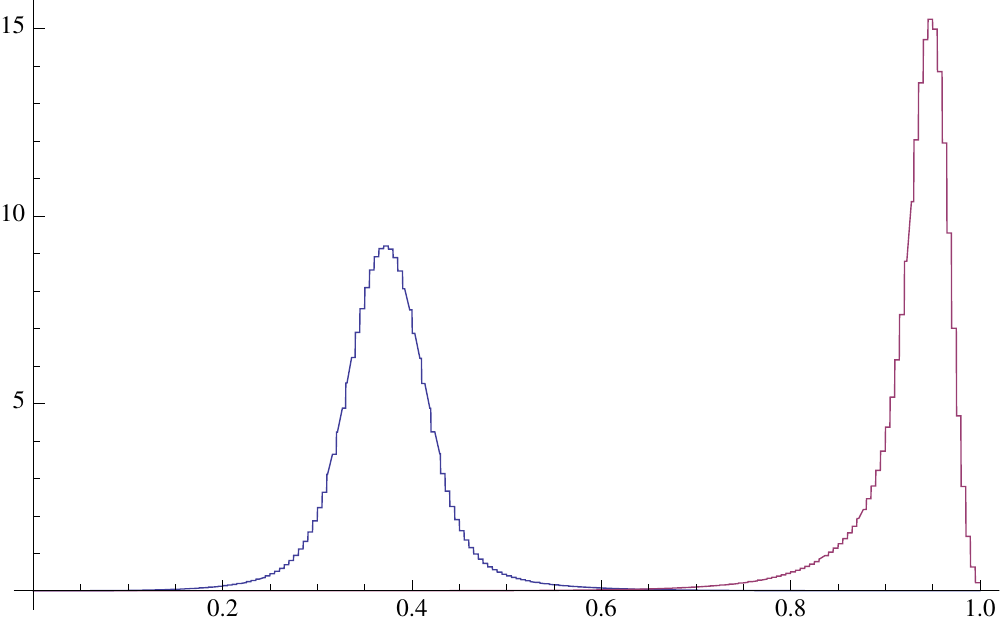}
\caption{\label{fig:MnatPDFRebits}Reconstructed--and linearly transformed to [0,1]--HS two-rebit probability distributions based on 200 moments of 
$|\rho| |\rho^{PT}|$ (blue, lower-peaked curve) and $|\rho^{PT}|$ (red, higher-peaked curve)}
\end{figure}
\begin{figure}
\includegraphics{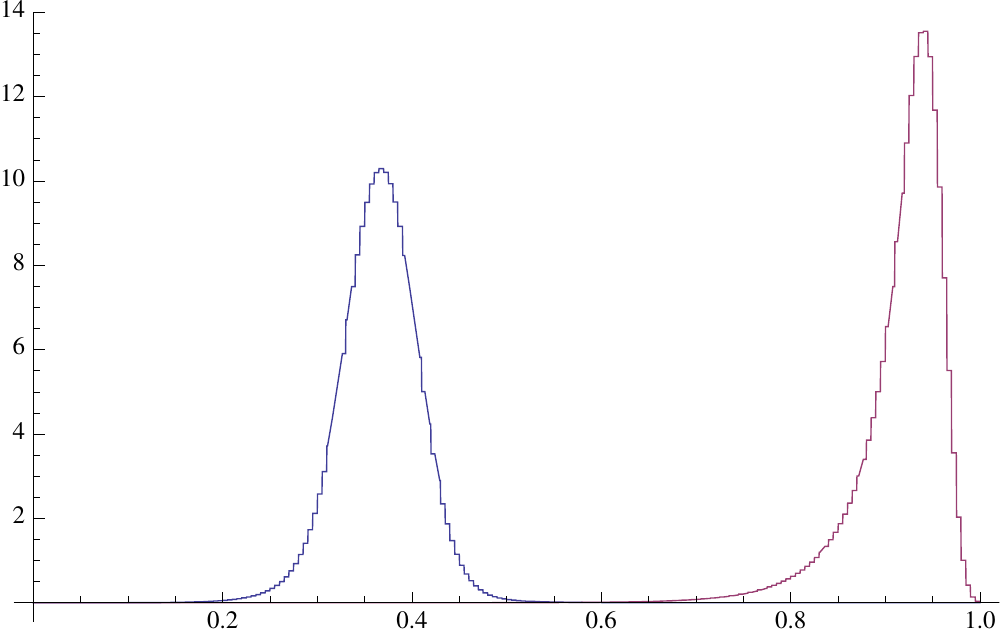}
\caption{\label{fig:MnatPDFQubits}Reconstructed--and linearly transformed to [0,1]--HS two-qubit probability distributions based on 200 moments of 
$|\rho| |\rho^{PT}|$ (blue, lower-peaked curve) and $|\rho^{PT}|$ (red, higher-peaked curve)}
\end{figure}
\subsection{$\alpha$ as a free parameter}
As an exercise of interest, let us consider the Dyson-index-like parameter 
$\alpha$ in sec.~\ref{ComplexConjectures}, with the values 
$\frac{1}{2}$ and 1 conjecturally corresponding to the two-rebit and two-qubit moments, respectively, as a {\it free/continuous} parameter (cf. \cite{MatrixModels}), and perform our standard separability probability calculations using the Provost algorithm 
\cite{Provost}--taking the same ranges as before for the 
determinantal moment variables. Based on ninety-six moments, we obtain 
Fig.~\ref{fig:FreeParameter}.
\begin{figure}
\includegraphics{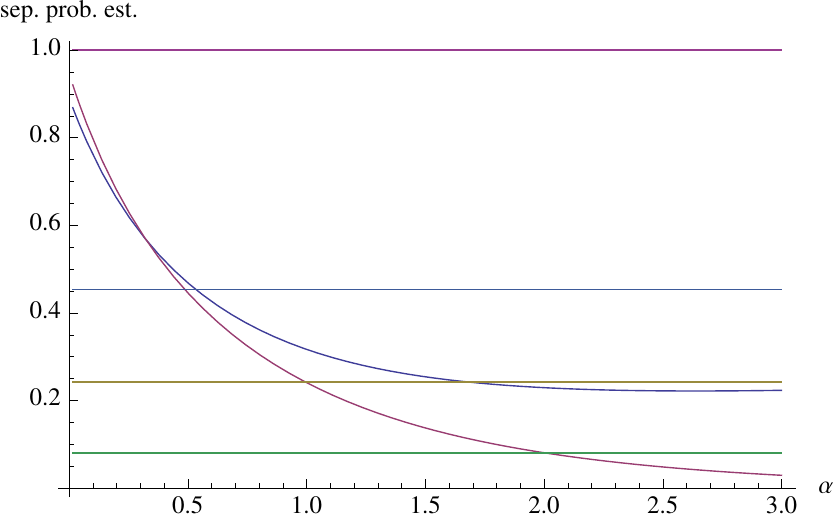}
\caption{\label{fig:FreeParameter}Separability probability estimates as a function of the parameter $\alpha$ (sec.~\ref{ComplexConjectures}). The upper curve is based on ninety-six moments of $|\rho| |\rho^{PT}|$, and the lower curve on ninety-six moments of $|\rho^{PT}|$. Also included as horizontal lines are the two-rebit ($\alpha=\frac{1}{2}$), two-qubit 
($\alpha=1$) and two-"quaterbit" ($\alpha=2$) and "classical" ($\alpha=0$) conjectures of  $\frac{29}{64} = 0.453125$,  $\frac{8}{33} \approx 0.242424$, $\frac{26}{323} \approx 0.080495$ and 1, respectively.}
\end{figure}
\subsection{$\alpha=2$ (quaternionic?)} \label{TWOQUATBITS}
In Fig.~\ref{fig:ListPlotQuatTwoMoments} we show--for the $\alpha =2$
(presumptively quaternionic) case (Appendix~\ref{ComplexConjectures})--the estimates obtained by application of the procedure of Provost \cite[eq. (15)]{Provost}
to the sequences of moments of $(|\rho| |\rho^{PT}|)^n$ (upper blue curve)
and $|\rho^{PT}|^n$ (lower red curve). (We use the term "presumptively", precisely because we have performed no explicit calculations--as we certainly have done in the two-rebit ($\alpha=\frac{1}{2}$) and two-qubit cases ($\alpha=1$)--involving
$4 \times 4$ {\it quaternionic} density matrices. We are, thus, proceeding under the assumption that we can extrapolate the formula of Dunkl to the case $\alpha=2$. Dunkl, however, has noted that his formula does agree with
that of Andai\cite[Thm. 4]{andai}, in the quaternionic case,  for the [univariate] moments of $|\rho|$ (cf. \cite{quartic}). Also, Dunkl has raised the issue of whether or not nonnegativity of the determinant of the partial transpose is equivalent to separability, as it is known to be in the two-rebit and two-qubit cases \cite{augusiak}.)
\begin{figure}
\includegraphics{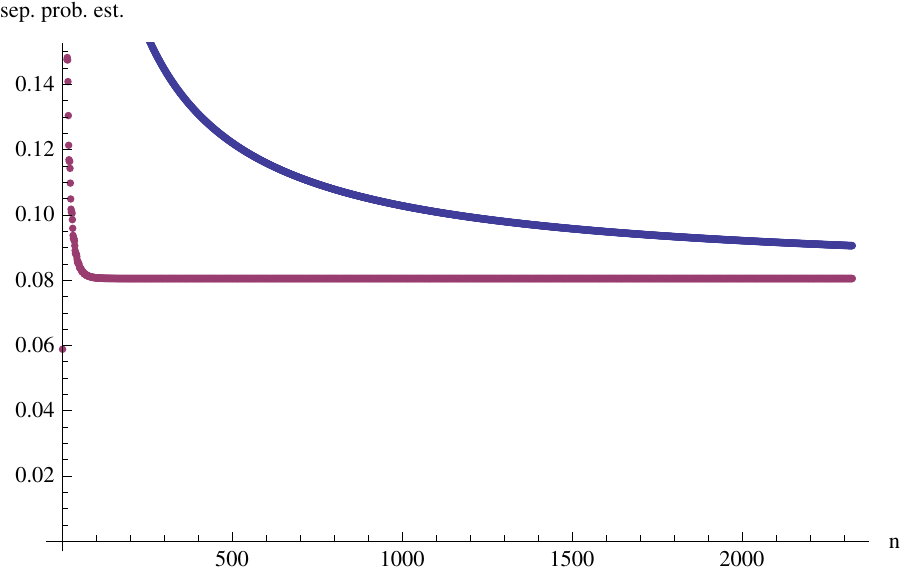}
\caption{\label{fig:ListPlotQuatTwoMoments}Two sets of estimates of the 
(quaternionic?) Hilbert-Schmidt separability probability. The upper (blue) decreasing curve is based on Dunkl's conjectured  
formulas--using $\alpha=2$--for the expected values of  $(|\rho| |\rho^{PT}|)^n$ and the lower (red) curve, similarly for   $|\rho^{PT}|^n$. 2,325 moments were employed.}
\end{figure}
The lower estimate based on 2,325 moments is 0.080495355 (which is 1.000000000049 times the corresponding estimate based on 2,324 moments). This 2,325-moment estimate can be truly remarkably well-fitted by 
the relatively simple fraction $\frac{26}{323} \approx 0.0804953560$.

In the framework of \cite{slater833}[sec. IX], the "scaling factor" used to  obtain the 
$\frac{26}{323}$ result would be $\frac{19136 \pi ^{12}}{152809335}$, 
where $ 19136 = 2^6 \times 13 \times 23$ and $152809335 =3^6 \times 5 \times 7 \times 53 \times 113$. (In these calculations, we took the  total HS quaternionic volume to be equal to the product of that volume given by Andai in \cite{andai} and the 
normalization factor of $2^{13}$ indicated there--thus, giving us the HS volume in the 
{\.Z}yczkowski-Sommmers framework \cite{szHS} that we have employed throughout.) For our two other conjectures, the associated scaling factors would be ($\alpha=\frac{1}{2}$, two-rebit) $\frac{145 \pi ^4}{128}$ and 
($\alpha=1$, two-qubit) $\frac{256 \pi ^6}{639}$. The associated HS 
{\it separable} volumes would, then,  be $ \frac{29 \pi ^4}{3870720}$, $\frac{2 \pi ^6}{7023641625}$, and $\frac{\pi ^{12}}{477802357101050231250}$, for the real, complex and quaternionic cases, respectively.
\subsection{$\alpha=4$ (octonionic?)} \label{TWOOCTBITS}
In Fig.~\ref{fig:ListPlotOctTwoMoments} we show--for the $\alpha =4$
(octonionic? (cf. \cite{adler,baez,dorje})) case--the estimates obtained by application of the procedure of Provost \cite[eq. (15)]{Provost}
to sequences of 2,125 moments of $(|\rho| |\rho^{PT}|)^n$ (upper blue curve)
and $|\rho^{PT}|^n$ (lower red curve). The fraction $\frac{760}{69903}  = \frac{2^3 \cdot 5 \cdot 19}{3^4 \cdot 863} \approx 0.0108722086$ is 0.9999999981  times as large as the estimated separability probability. 
Convergence is comparatively very strong in this instance, and definitely seems to improve, in general, as the Dyson-index-like parameter $\alpha$ increases.
\begin{figure}
\includegraphics{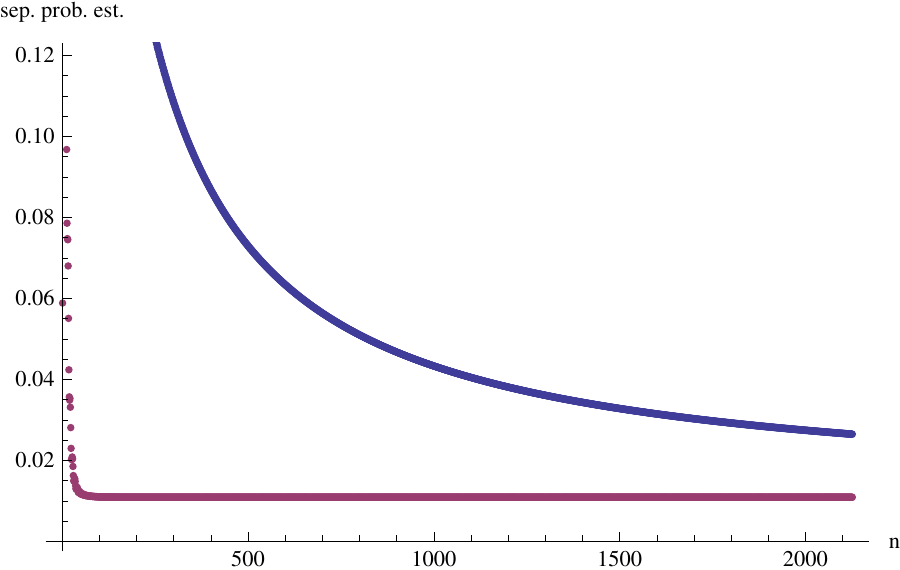}
\caption{\label{fig:ListPlotOctTwoMoments}Two sets of estimates of the 
(octonionic?) Hilbert-Schmidt separability probability. The upper (blue) decreasing curve is based on Dunkl's conjectured  
formulas--using $\alpha=4$--for the expected values of  $(|\rho| |\rho^{PT}|)^n$ and the lower (red) curve, similarly for   $|\rho^{PT}|^n$. The true value appears to be constrained to lie within [0.0108722086, 0.0264396063]. 2,125 moments were employed.}
\end{figure}
\subsection{$\alpha=0$ (classical?)} \label{TWONULLBITS}
If we set $\alpha=0$ in (\ref{nequalk}) for the (mixed-moments) 
case $n=k$, we obtain the simplification
\begin{equation} \label{nequalkSimple}
\left\langle \left\vert \rho\right\vert ^{n}\left\vert \rho^{PT}\right\vert
^{n}\right\rangle =
\frac{4096^{-n} \Gamma (2 n+1)^3}{\left(\frac{3}{2}\right)_{2 n}
   \left(\frac{5}{2}\right)_{4 n}}.
\end{equation}
In Fig.~\ref{fig:ListPlotOctNULLMoments}, we plot our standard pair of two estimates (although now the roles of upper and lower curves are reversed).
It appears that there is convergence to 1, that is, $\alpha=0$ corresponds, in some sense, to a {\it classical} scenario, in which no entanglement is present.
\begin{figure}
\includegraphics{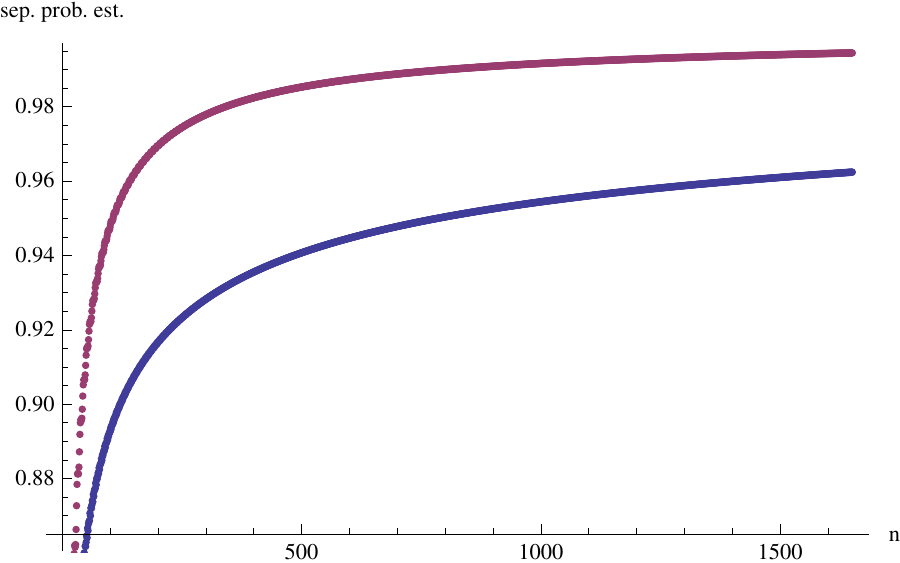}
\caption{\label{fig:ListPlotOctNULLMoments}Two sets of estimates of the 
(classical?)  Hilbert-Schmidt separability probability, using $\alpha=0$. The upper (red) decreasing curve is based on Dunkl's conjectured  
formula (\ref{nequalzero}) for the expected values of  
$|\rho^{PT}|^n$--and the lower (blue) curve, similarly for $(|\rho| 
|\rho^{PT}|)^n$--given by the simplified formula (\ref{nequalkSimple}). 1,650 moments were employed, with the last pair of estimates being $\{0.96238936,0.99445741\}$.}
\end{figure}
In regard to setting  $\alpha=0$, Dunkl commented that doing so "assigns measure zero to the off-diagonal entries of the Cholesky factor. The determinant and PT-determinant are identical as far as the measure is concerned, and the probability distribution is the same as that of the product $t_1 t_2 t_3 t_4$ on the simplex in 3-space ($t_1,t_2,t_3 \geq 0, t_4=1-t_1-t_2-t_3$ and $t_4 \geq 0$)."  His "attempt to reconstruct the underlying probability distribution yields an inelegant integral of a hypergeometric series".
\subsection{Other values of $\alpha$}
We also have conducted Legendre-polynomial reconstruction  analyses for a number of other values of $\alpha$, which we summarize in the form
(cf. Fig.~\ref{fig:FreeParameter})
\begin{equation}
\left(
\begin{array}{ccc}
 \frac{1}{4} & 850 & \{0.64744667, 0.63955009\} \\
 \frac{3}{4} & 525 & \{0.34299437,0.32784144\} \\
 \frac{3}{2} & 1600 & \{0.13756171,0.14950325\} \\
 3 & 1075 & \{0.029008076,0.055230359\} \\
 8 & 850 & \{0.00025439139228,0.055713576\}
\end{array}
\right) .
\end{equation}
The first two columns give the value of $\alpha$ and the number of moments employed, and the last, the confining interval for the associated separability probabilities, with the first value being based on the moments of $|\rho^{PT}|$ and the second, on the moments of 
$|\rho| |\rho^{PT}|$. Convergence of the probability-distribution reconstruction algorithm, based on the moments of $|\rho^{PT}|$, appears to greatly increase as $\alpha$ increases. (An extremely close fractional fit to the lower bound for $\alpha=8$ is $\frac{81}{318407} \approx 
0.00025439139215$.)
\subsection{Specialized lower-dimensional ("non-generic") cases} \label{specialized}
In \cite[sec. II.A]{slater833}, we considered classes of $4 \times 4$
real, complex and quaternionic density matrices, where--as usual--the diagonal entries were allowed to take values in the 3-simplex, but now five of the six pairs of  off-diagonal entries were nullified, leaving only the (2,3) and (3,2)-pair as free.
(The associated separability probabilities were found to be $\frac{3 \pi}{16},\frac{1}{3}$ and $\frac{1}{10}$.) Dunkl (App.~\ref{fourparameter}) has now been able to {\it prove} formulas for the bivariate moments in these specialized scenarios.
\section{Hilbert-Schmidt and Bures probability distributions over 
$|\rho|$}
In the course of this work, Charles Dunkl further communicated to us a 
result (following his joint work with K. {\.Z}yzckowski reported in \cite{dunkl}, where 
"the machinery for producing densities from moments of Pochhammer type" was developed)
giving the univariate probability distribution over $t \in [0,1]$ that reproduces the Hilbert-Schmidt moments of $t=2^8 |\rho|$, where $\rho$ is a generic two-rebit density matrix. (If we set $n=0$ in our general 
[bivariate] determinantal moment framework above, we obtain the [univariate] moments of $|\rho|$.) This probability distribution took the 
form (cf. \cite{csz}[eq. (4.3)])
\begin{equation} \label{marginaldistribution}
\frac{63}{8} \left(\sqrt{1-\sqrt{t}} \left(-8 t-9 \sqrt{t}+2\right)+15 t
   \log \left(\sqrt{1-\sqrt{t}}+1\right)-\frac{15}{4} t \log (t)\right)
\end{equation}
(see Appendix~\ref{DensityHS} below for further details).
At the suggestion of the author, Dunkl was also able to derive, in similar fashion, the Bures metric \cite{szBures,ingemarkarol} counterpart of this Hilbert-Schmidt  result (\ref{marginaldistribution}). It took the form (Appendix~\ref{DensityBures})
\begin{equation} \label{marginaldistributionBures}
\frac{-4 \sqrt{\sqrt{t}-t} \left(2 \sqrt{t}+13\right)+3 \pi  \left(4
   \sqrt{t}+1\right)+2 \left(12 \sqrt{t}+3\right) \sin ^{-1}\left(1-2
   \sqrt{t}\right)}{\pi  \sqrt{t}}.
\end{equation}
In Fig.~\ref{fig:twoprobs} we display these two (Hilbert-Schmidt and Bures) probability distributions.
\begin{figure} 
\includegraphics{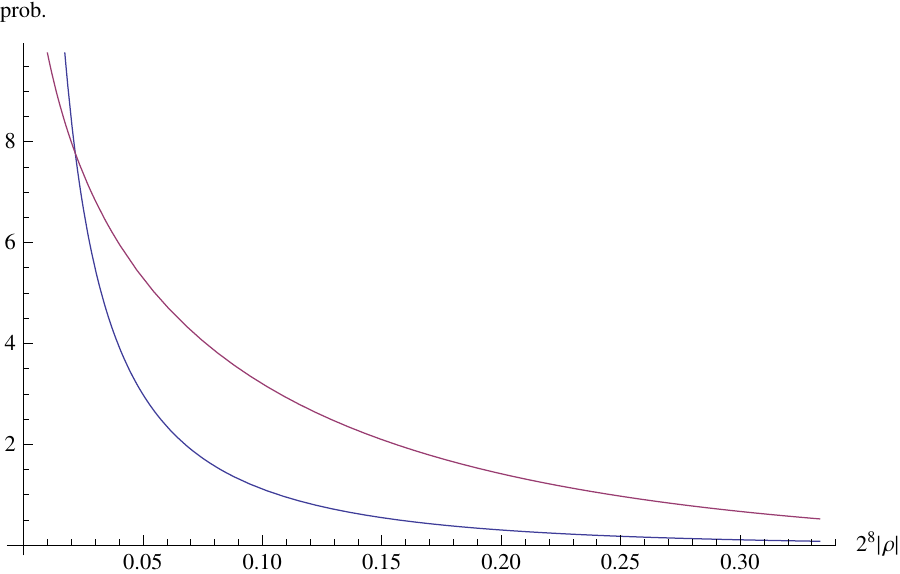}
\caption{\label{fig:twoprobs}Probability distributions (\ref{marginaldistribution}) 
and (\ref{marginaldistributionBures}) over $t= 2^8 |\rho|$ ($t \in [0,1]$). The Hilbert-Schmidt (red) curve dominates the Bures curve above $t = 0.021702$.}
\end{figure}
\section{Discussion}
\subsection{Background}
A basic linear-algebraic criterion that a Hermitian matrix be nonnegative-definite, that is have all its eigenvalues nonnegative, is that all its principal minors be nonnegative. In \cite{advances}, we were able to implement this criterion, in part, making use of the $3 \times 3$ minors, establishing  thereby that the Hilbert-Schmidt probability a generic two-rebit system is separable is bounded above by $\frac{1129}{2100} \approx 0.537619$. 
(The {\it absolute} separability probability of 
$\frac{6928-2205 \pi }{2^{9/2}} \approx 0.0348338$ provided the best exact {\it lower} bound established in this specific setting \cite{advances}, it appeared. The set of absolutely separable two-qubit states are described in Figs. 1-5 in \cite{JMP2008} (cf. \cite{kuszyczkowski,thirring,zanardi}). No immediate application of the moment-based approach adopted in this study to the description of the absolutely separable states is apparent.) That study \cite{advances} was a continuation of a series of papers of ours 
(including \cite{slaterJGP,slaterqip,slaterA,slaterC,slaterPRA,slaterPRA2,slaterJGP2,pbsCanosa,slater833}) in which we examined the separability probability question--for the Hilbert-Schmidt as well as various monotone (such as the Bures) metrics--from a variety of mathematical perspectives, employing a number of density-matrix parameterizations. A major motivation in undertaking the moment-related 
analyses reported above was to further sharpen
our separability probability estimates, perhaps even being able to arrive at an estimate accurate to several decimal places, and possibly obtain thereby convincing evidence for a particular true value.

Despite the considerable computational efforts expended in calculating high-order moments, the goal of high accuracy nevertheless appeared  remote--that is, until the apparent advances of Dunkl  
(Appendix~\ref{appDunkl}) that we have sought to subsequently exploit above. This 
somewhat pessimistic viewpoint had been based on a continuing series of attempts by us--using a wide variety of probability-density reconstruction methodologies--to isolate the two-rebit separability probability on the basis of the initially computed (limited number of) thirteen moments.
As an example (cf. sec.~\ref{Gaussian}), use of the nonparametric procedure of Mnatsakanov 
\cite{mnatsakanov2}, yielded HS generic two-rebit separability probability estimates of 0.4582596, 0.42970496 and 0.40321291 based on the first eleventh, twelfth and thirteen moments of $|\rho^{PT}|$ (sec.~\ref{PPTmoments}), so, no convergence was apparent, at least,  with these few moments. The corresponding estimates were 0.5414052, 0.3923661 and 0.4792091 based on eleventh, twelfth and thirteen moments of $|\rho| |\rho^{PT}|$ (sec.~\ref{EQUALmoments}). Use of the first ten moments in a certain
maximum-entropy reconstruction methodology \cite{Biswas} gave an estimate of 0.409858.
Additionally, incorporation of the first twelve moments into an
adaptive spline-based algorithm \cite{john2} gave 0.4502338.  The semiparametric Legendre-polynomial-based reconstruction approach of Provost \cite{Provost}--our chief computational procedure in the main body of this paper--gave
estimates of 0.3856787 and 0.4846628 based on the first thirteen moments
of $|\rho^{PT}|$ and $|\rho| |\rho^{PT}|$, respectively.

We had, thus,  before the general formula of Dunkl, encountered evident difficulties in ascertaining to {\it high} accuracy the values of separability probabilities. These difficulties, it seemed, perhaps 
manifested the {\it  NP}-hardness of the problem of distinguishing separable quantum states from entangled ones  \cite{gurvits,ioannou,gharibian}.
As possible evidence for such a contention, if one knew {\it all} the generic HS two-rebit moments of $|\rho^{PT}|$, then presumably one could determine the associated separability probability to arbitrarily high accuracy. But to know all these moments, it appeared that one would have to know an indefinitely large number of the functions $C_{3 n-i}$ 
((\ref{specfunct1})-(\ref{specfunct2})), from which the needed constant terms could be extracted. In the apparent absence of a generating rule for
these increasingly high-order functions (but see  
Appendix~\ref{appDunkl}), an indefinitely large amount of computation appeared to be required. ("Although [quantum entanglement] is usually fragile to the environment, it is robust against conceptual and mathematical tools, the task of which is to decipher its rich structure" \cite[p. 865]{family}.) "In \cite[sec. II.B]{HSorthogonal}, an earlier study of ours of the moments for two-rebit systems, we encountered a somewhat analogous rather intractable 
state-of-affairs, employing the Bloore 
(correlation-coefficient) parameterization of density matrices (and not the Cholesky decomposition parameterization, as in this study). There, a general formula for the denominators of certain important "intermediate functions" could be discerned, but only explicit results obtained for an {\it initial} set ($m=2, 4, 6,\dots,16$) of the corresponding numerators. So, higher-order moments--and, thus, high accuracy--appeared out of reach there (but certainly in light of the 
apparent progress--but not yet rigorously established--of Dunkl, the matters there might also be readdressed).
\subsection{Results}
In this paper, we have advanced four specific conjectures 
($\alpha = 0, \frac{1}{2}, 1, 2$) (Fig.~\ref{fig:FreeParameter}). The reader might have been somewhat skeptical of our strong predisposition to conjecture {\it rational} values for the various separability probabilities under consideration. 
A basis for this inclination had been established in \cite{slater833}, where a pattern of {\it rational} separability probabilities appeared through the application of {\it exact} methods to lower-dimensional non-generic (but more easily computed) quantum scenarios (sec.~\ref{specialized}).

In regard to the conjecture \cite[sec. IX.B]{slater833} that the Hilbert-Schmidt separability probability of generic (15-dimensional) two-qubit systems 
is $\frac{8}{33}$, K. {\.Z}yczkowski informally wrote: 
"It would be amazing if such a simple number occurs to be true!
I wonder then if it is likely that this result may be derived analytically
(by a clever integration), or perhaps even 'guessed' from some
symmetry arguments [which are still missing]". From the author's viewpoint, perhaps one of the chief hurdles here is simply the 
exceptionally high-dimensionality and {\it quartic} (separability) constraints that need to be addressed in any integration
("clever" or otherwise). Possibly with the advent of more powerful
symbolic (quantum?) computational systems, this obstacle might be directly overcome.
Also, in terms of symmetry principles, the (Keplerian) concept of "stella octangula" \cite{aravind,ericsson} has proved useful in studying separability, and might conceivably do so (in some higher-dimensional realization) in the future. Certain interesting aspects of convexity were applied in \cite{sbz} to obtain theorems pertaining to Hilbert-Schmidt separability probabilities.

The general formulas of Dunkl remain formally unproven. However, our confidence in their validity is certainly enhanced by the reasonableness and non-anomalous behavior (Figs.~\ref{fig:ListPlotTwoMoments}, \ref{fig:ListPlotComplexTwoMoments}, \ref{fig:MnatPDFRebits},  \ref{fig:MnatPDFQubits}, \ref{fig:ListPlotQuatTwoMoments},  \ref{fig:ListPlotOctTwoMoments}) of our various (separability) probability estimation procedures, for various values of $\alpha$, which rely upon them. If the formulas did not, in fact, yield genuine moments of probability distributions,  we would certainly expect that to be manifested, in some overt manner (negative probabilities, probabilities greater than unity,
non-convergent behavior,\ldots) in our reconstruction efforts.

It is interesting to note that of our three basic (two-rebit, 
-qubit, -"quaterbit" \cite{batle2}) separability probability conjectures--$\frac{29}{64}, \frac{8}{33}, \frac{26}{323}$--the two-qubit is the simplest, in the sense of having the smallest denominator (and numerator). The two-qubit systems exist conceptually in the framework of  (standard/conventional/phenomenological) {\it complex} quantum mechanics \cite[sec. 2]{quartic,baez}.

A further observation is that although in random matrix theory, a (Dyson-index) parameter
$\beta =1$ (the dimension of the corresponding {\it division algebra} \cite{baez}) is typically assigned to the real systems, in the (Cholesky decomposition-based) analysis of Dunkl (App.~\ref{appDunkl}), the use, instead, of 
$\alpha =\frac{1}{2}$ appears to be natural--since one-halves repeatedly arise in the integration over the real sphere in $\mathcal{R}^{10}$.

Knowledge of all the moments of $|\rho^{PT}|$ and $|\rho| |\rho^{PT}|$
theoretically determines the {\it complete} probability distributions of these two variables (since the ranges of these two variables are bounded).
In some sense, this constitutes more information than it might seem one should require to determine the {\it single} (separability) probability of primary, motivational interest \cite{ZHSL}. So, if at some point in time, the separability probability questions can be resolved by some more {\it direct} methods, than it may appear that the analytical moment-based approach pursued here was more than was, in fact, truly required for the task at hand. Nevertheless, in the interim, this approach has clearly greatly advanced our knowledge of the ranges within which the separability probabilities must lie--even if not helping to pinpoint their conjectured exact (simple rational) values.
\subsection{Bures analyses}
In a naive exercise, we investigated whether or not the bivariate moment formulas presented here might further hold--at least up to proportionality--if one were to simply replace the expectation with respect to the Hilbert-Schmidt metric in them by expectation with respect to the Bures (minimal monotone) metric \cite{szBures,ingemarkarol,andai,slaterC,slaterJGP}. However, such a possible relationship appeared to be quite emphatically ruled out, at least with the one specific example, formula (\ref{thirdnew}) above, we numerically studied in these regards.

In \cite[eq. (16)]{slaterJGP} we had--based on extensive quasi-Monte Carlo numerical integrations--advanced the hypothesis that the two-qubit {\it Bures} separability probability took the form (with the "silver mean", 
$\sigma_{Ag} = \sqrt{2}-1$)
\begin{equation} \label{silverconjecture}
P^{sep}_{Bures} = \frac{1680 \sigma_{Ag}}{\pi^8} \approx 0.07333893767 
\end{equation}
(which we do note is obviously {\it irrational}--in contrast to our Hilbert-Schmidt conjectures).
We have recently begun to reexamine the results of that 2005 study, particularly in light of the later (2009) development, making use of Ginibre ensembles, of a 
"simple and efficient algorithm to generate at random, density matrices distributed according to the Bures measure" \cite{osipov} (cf. \cite[eq. (22)]{generating}). 
In an {\it ongoing} calculation, employing extended-precision independent normal random variables, we have obtained (using the normal approximation to the binomial distribution)--based on 281,350,000 realizations (20,627,508 being separable, giving a probability of 0.0733162)--a $95\%$ confidence interval $\{0.07328572, 0.07334664\}$. We note that this interval {\it does contain} the conjectured value  (\ref{silverconjecture}) for the true Bures two-qubit separability probability.
(Consistently with these analyses,  if we introduce our Hilbert-Schmidt two-qubit separability-probability conjecture of $\frac{8}{33}$ into the inequality of Ye \cite[mid. p. 7]{ye}, we obtain  0.00373882 as a lower bound on the {\it Bures} two-qubit separability probability. Application of the very next inequality of Ye appears to yield 599089., obviously greater than 1, as an {\it upper} bound on this {\it probability}.)

\appendix
\section{Two-rebit Hilbert-Schmidt moments $\left\langle |\rho^{PT}|^{n} \right\rangle_{2-rebit/HS}$, $n=1,\ldots,13$} \label{PPTmoments}
\begin{equation}
\left(
\begin{array}{ccc}
 1 & -\frac{1}{858} & -0.0011655 \\
 2 & \frac{27}{2489344} & 0.0000108462 \\
 3 & -\frac{8363}{66216550400} & -1.26298\times 10^{-7} \\
 4 & \frac{21859}{10443295948800} & 2.09311\times 10^{-9} \\
 5 & -\frac{23071}{539633583390720} & -4.27531\times 10^{-11} \\
 6 & \frac{3317321}{3253917653076541440} & 1.01949\times 10^{-12} \\
 7 & -\frac{419856257}{15366774022001834065920} & -2.73223\times 10^{-14}
   \\
 8 & \frac{16945249}{21117403549591928832000} & 8.02431\times 10^{-16} \\
 9 & -\frac{6102620963}{240565904621616585139814400} & -2.53678\times
   10^{-17} \\
 10 & \frac{87816716413}{103068223454742370906999357440} & 8.52025\times
   10^{-19} \\
 11 & -\frac{7685831825319}{255310031843279606667374504181760} &
   -3.01039\times 10^{-20} \\
 12 & \frac{23559692226221}{21217623285399369347467109090721792} &
   1.11038\times 10^{-21} \\
 13 & -\frac{31283325154283}{736092406055063912488279599166259200} &
   -4.24992\times 10^{-23}
\end{array}
\right)
\end{equation}
\section{Two-rebit Hilbert-Schmidt moments $\left\langle (|\rho|  |\rho^{PT}|)^{n} \right\rangle_{2-rebit/HS}$, $n=1,\ldots,13$} \label{EQUALmoments}
\begin{equation}
\left(
\begin{array}{ccc}
 1 & 0 & 0. \\
 2 & \frac{7}{5696343244800} & 1.22886\times 10^{-12} \\
 3 & \frac{1}{677899511057612800} & 1.47514\times 10^{-18} \\
 4 & \frac{1}{45973294808920227840000} & 2.17518\times 10^{-23} \\
 5 & \frac{1}{11662680803407302839532257280} & 8.57436\times 10^{-29} \\
 6 & \frac{3929}{4158654163938276392103553381781471232} & 9.44777\times
   10^{-34} \\
 7 & \frac{1}{158158366213274948625327048295175946240} & 6.32278\times
   10^{-39} \\
 8 & \frac{71527}{1091771390479438557169317171313498708778365747200} &
   6.55146\times 10^{-44} \\
 9 & \frac{4847}{8524774835462825812953111833131999123882778862551040} &
   5.68578\times 10^{-49} \\
 10 &
   \frac{2637}{441859421690475898778224458156196857558486829112995348480}
   & 5.96796\times 10^{-54} \\
 11 &
   \frac{1}{16833241044745336849504728327369136893812113975649318731776}
   & 5.94063\times 10^{-59} \\
 12 &
   \frac{66838003}{103562821755098721107694750210986399334006111231977898
   090691403395891200} & 6.45386\times 10^{-64} \\
 13 &
   \frac{55601}{799197847439412434413767694576364829625113476005862347670
   3807122125619200} & 6.9571\times 10^{-69}
\end{array}
\right)
\end{equation}
\section{Moments of $|\rho^{PT}|^{n}, n=1,\ldots,10$, for minimally degenerate pairs of rebits} \label{AREA}
\begin{equation}
\left(
\begin{array}{ccc}
 1 & -\frac{5}{2376} & -0.00210438 \\
 2 & \frac{7}{380160} & 0.0000184133 \\
 3 & -\frac{9}{34777600} & -2.58787\times 10^{-7} \\
 4 & \frac{443}{89942261760} & 4.92538\times 10^{-9} \\
 5 & -\frac{461}{4032782401536} & -1.14313\times 10^{-10} \\
 6 & \frac{5455}{1785064543223808} & 3.05591\times 10^{-12} \\
 7 & -\frac{631}{6948198442598400} & -9.08149\times 10^{-14} \\
 8 & \frac{474017}{161763811601154048000} & 2.9303\times 10^{-15} \\
 9 & -\frac{4003573}{39645007353595350220800} & -1.00986\times 10^{-16}
   \\
 10 & \frac{3397}{924892257224239349760} & 3.67286\times 10^{-18}
\end{array}
\right)
\end{equation}
\section{Two-rebit and two-qubit moments}  \label{appDunkl}

\textbf{Charles F. Dunkl}\footnote{Department of Mathematics, University of
Virginia, Charlottesville VA, 22904-4137}\footnote{Email: cfd5z@virginia.edu}

Let $\Omega$ denote the set of $4$-by-$4$ (symmetric) real positive definite
matrices, and let $\Omega_{1}$ denote the matrices of trace one in $\Omega$.
Recall $\left\langle X\right\rangle $ denotes the expectation of the random
variable $X$, with the associated probability density being implicit from the
text. Furthermore $\left\vert \rho\right\vert $ denotes det $\rho$.

\subsection{Construction of density functions}

We describe the tools used to determine densities whose moment sequence is
given in Pochhammer form. Here we restrict to densities supported on $\left[
0,1\right]  $. Let $f\left(  x\right)  $ be defined on $0\leq x\leq1$, such
that $f\left(  x\right)  \geq0$, $f$ is continuous on $0<x<1$ and $\int
_{0}^{1}f\left(  x\right)  dx=1$. There is an associated random variable $X$,
with $\Pr\left\{  a<X<b\right\}  =\int_{a}^{b}f\left(  x\right)  dx$. The
moment sequence is $\left\langle X^{n}\right\rangle =\int_{0}^{1}x^{n}f\left(
x\right)  dx,n=0,1,2\ldots$. Observe that the moment sequence uniquely defines
the density because the support is a bounded interval.

First we consider a beta-type distribution: let $\alpha,\beta>0$, and
\begin{align}
f\left(  x\right)   &  =\frac{1}{B\left(  \alpha,\beta\right)  }x^{\alpha
-1}\left(  1-x\right)  ^{\beta-1},0<x<1,\label{betaX}\\
\int_{0}^{1}x^{n}f\left(  x\right)  dx  &  =\frac{\left(  \alpha\right)  _{n}%
}{\left(  \alpha+\beta\right)  _{n}},n=0,1,2,\ldots.\nonumber
\end{align}
(Recall $B\left(  \alpha,\beta\right)  =\frac{\Gamma\left(  \alpha\right)
\Gamma\left(  \beta\right)  }{\Gamma\left(  \alpha+\beta\right)  }$.) This
uses the identity $\Gamma\left(  \alpha+n\right)  /\Gamma\left(
\alpha\right)  =\left(  \alpha\right)  _{n}:=\prod_{i=1}^{n}\left(
\alpha+i-1\right)  $, the Pochhammer symbol.

\begin{lemma}
\label{X1X2}Suppose $X_{1},X_{2}$ are independent random variables on $\left[
0,1\right]  $ with densities $f_{i}$, $i=1,2$. Then the density for
$X_{1}X_{2}$ is%
\[
f\left(  x\right)  :=\int_{x}^{1}f_{1}\left(  t\right)  f_{2}\left(  \frac
{x}{t}\right)  \frac{1}{t}dt.
\]
If the moments of $X_{1},X_{2}$ are $\mu_{n}^{\left(  i\right)  }=\left\langle
X_{i}^{n}\right\rangle =\int_{0}^{1}x^{n}f_{i}\left(  x\right)  dx$ then
$\left\langle X_{1}^{n}X_{2}^{n}\right\rangle =\mu_{n}^{\left(  1\right)  }%
\mu_{n}^{\left(  2\right)  },n=0,1,2\ldots$. , that is,
\[
\int_{0}^{1}x^{n}f\left(  x\right)  dx=\mu_{n}^{\left(  1\right)  }\mu
_{n}^{\left(  2\right)  },n=0,1,2,\ldots.
\]

\end{lemma}

The Lemma was stated and used in \cite[p.123521-20]{dunkl}. Also we use the
duplication formulae for Pochhammer symbols:%
\begin{align*}
\left(  a\right)  _{2n}  &  =2^{2n}\left(  \frac{a}{2}\right)  _{n}\left(
\frac{a+1}{2}\right)  _{n},\\
\left(  2n\right)  !  &  =\left(  1\right)  _{2n}=2^{2n}n!\left(  \frac{1}%
{2}\right)  _{n},\\
\left(  2n+1\right)  !  &  =\left(  2\right)  _{2n}=2^{2n}n!\left(  \frac
{3}{2}\right)  _{n}.
\end{align*}

\subsection{Density of the determinant under the Hilbert-Schmidt metric} \label{DensityHS}

The $10$-dimensional cone $\Omega$ is equipped with the measure $\prod_{1\leq
i\leq j\leq4}d\rho_{ij}$ (where $\rho=\left(  \rho_{ij}\right)  _{i,j=1}^{4}$ is the
generic matrix). The probability distribution on $\Omega_{1}$ is the
($9$-dimensional) restriction of this measure.

The following lemma applies to $N$-by-$N$ positive-definite matrices for any
$N=2,3,\ldots$. Each such matrix $\rho$ has a Cholesky decomposition:%
\[
\rho=C^{t}C,
\]
where $C$ is upper triangular with entries $c_{ij},$ $c_{ij}=0$ for $i>j$ and
$c_{ii}\geq0$ for all $i$. The entries of $\rho$ are $\rho_{ij}=\sum_{k=1}%
^{N}c_{ki}c_{kj}=\sum_{k=1}^{\min\left(  i,j\right)  }c_{ki}c_{kj}$. Consider
the Jacobian matrix $\frac{\partial \rho}{\partial c}$ where the dependent
variables are $\rho_{ij},i\leq j$.

\begin{lemma}
Suppose $\rho=C^{t}C$ then%
\[
\left\vert \det\frac{\partial \rho}{\partial c}\right\vert =2^{N}\prod_{i=1}%
^{N}c_{ii}^{N+1-i}.
\]

\end{lemma}

\begin{proof}
We use the simple fact: suppose $y_{i}=f_{i}\left(  x_{1},x_{2},\ldots
,x_{i}\right)  $, $1\leq i\leq N$ then the matrix $\left(  \frac{\partial
y_{i}}{\partial x_{j}}\right)  $ is lower-triangular ($0$ for $j>i$) and
$\det\left(  \frac{\partial y_{i}}{\partial x_{j}}\right)  =\prod_{i=1}%
^{N}\frac{\partial y_{i}}{\partial x_{i}}$. Now order the (independent)
variables: $c_{11},c_{12},\ldots,c_{1N},c_{22},\ldots$ , $c_{2N}$, $c_{33}$,
\ldots$c_{N-1,N-1},c_{N-1,N},c_{NN}$. For $i\leq j$, $\rho_{ij}=\sum_{k=1}%
^{i-1}c_{ki}c_{kj}+c_{ii}c_{ij}$and thus
\[
\left\vert \det\frac{\partial \rho}{\partial c}\right\vert =\prod_{i=1}^{N}%
\prod_{j=i}^{N}\frac{\partial \rho_{ij}}{\partial c_{ij}}=\prod_{i=1}^{N}\left(
2c_{ii}^{N-i+1}\right)  .
\]

\end{proof}

Now set $N=4$. The pre-image $S$ of $\Omega_{1}$ (for the map $C\mapsto
C^{t}C$) is a modified octant of the unit sphere in $\mathbb{R}^{10}$, because
$Tr\left(  C^{t}C\right)  =\sum_{1\leq i\leq j\leq4}c_{ij}^{2}$. Recall the
condition $c_{ii}\geq0$, but the other entries can have arbitrary signs. The
surface measure $dm\left(  C\right)  $ on $S$ is essentially a Dirichlet
measure: consider a monomial on $S$, that is,%
\[
f\left(  C\right)  :=\prod_{1\leq i\leq j\leq4}c_{ij}^{n_{ij}},
\]
then

\begin{enumerate}
\item if $n_{ij}$ is odd for some $i<j$ then $\int_{S}f\left(  C\right)
dm\left(  C\right)  =0$,

\item if $n_{ij}$ is even for each $i<j$ then%
\[
\int_{S}f\left(  C\right)  dm\left(  C\right)  =\frac{\Gamma\left(  5\right)
}{\Gamma\left(  \frac{1}{2}\right)  ^{10}}\frac{1}{\Gamma\left(  5+\frac{1}%
{2}\sum_{1\leq i\leq j\leq4}n_{ij}\right)  }\prod_{1\leq i\leq j\leq4}%
\Gamma\left(  \frac{1}{2}+\frac{n_{ij}}{2}\right)  ,
\]

\item if $n_{ij}$ is even for each $i\leq j$, and $N:=\sum_{1\leq i\leq
j\leq4}n_{ij}$ then%
\[
\int_{S}f\left(  C\right)  dm\left(  C\right)  =\frac{1}{\left(  5\right)
_{N}}\prod_{1\leq i\leq j\leq4}\left(  \frac{1}{2}\right)  _{n_{ij}/2}.
\]

\end{enumerate}

In our usage either case 1 or case 3 applies. Combining the Jacobian and the
fact $\left\vert \rho\right\vert =c_{11}^{2}c_{22}^{2}c_{33}^{2}c_{44}^{2}$ we
obtain (for normalized measure, that is $\gamma\int_{S}\det\frac{\partial
\rho}{\partial c}dm\left(  C\right)  =1$), $k=0,1,2,\ldots$:

\begin{enumerate}
\item if $n_{ij}$ is odd for some $i<j$ then $\int_{S}\left\vert
\rho\right\vert ^{k}f\left(  C\right)  \det\frac{\partial\rho}{\partial
c}dm\left(  C\right)  =0$,

\item if $n_{ij}$ is even for each $i\leq j$, and $N:=\sum_{1\leq i\leq
j\leq4}n_{ij}$ then%
\begin{gather*}
\gamma\int_{S}\left\vert \rho\right\vert ^{k}f\left(  C\right)  \det
\frac{\partial \rho}{\partial c}dm\left(  C\right)  =\frac{1}{\left(  10\right)
_{4k+N/2}}\\
\times\left(  \frac{5}{2}\right)  _{k+n_{11}/2}\left(  2\right)  _{k+n_{22}%
/2}\left(  \frac{3}{2}\right)  _{k+n_{33}/2}\left(  1\right)  _{k+n_{44}%
/2}\prod_{1\leq i<j\leq4}\left(  \frac{1}{2}\right)  _{n_{ij}/2}.
\end{gather*}

\end{enumerate}

The special case $f\left(  C\right)  =1$ provides the moments of the random
variable $\left\langle \rho\right\rangle $; indeed
\[
\gamma\int_{S}\left\vert \rho\right\vert ^{k}\det\frac{\partial \rho}{\partial
c}dm\left(  C\right)  =\frac{\left(  \frac{5}{2}\right)  _{k}\left(  2\right)
_{k}\left(  \frac{3}{2}\right)  _{k}\left(  1\right)  _{k}}{\left(  10\right)
_{4k\ }}.
\]
We know the range of $\left\vert \rho\right\vert $ is$\left[  0,\frac{1}%
{256}\right]  $ (the maximum is achieved at $\rho=\frac{1}{4}I$); to use the
previous results consider $X=2^{8}\left\vert \rho\right\vert $. Then%
\begin{align*}
\left\langle X^{n}\right\rangle  &  =2^{8n}\frac{\left(  \frac{5}{2}\right)
_{n}\left(  2\right)  _{n}\left(  \frac{3}{2}\right)  _{n}\left(  1\right)
_{n}}{2^{4n}\left(  5\right)  _{2n}\left(  \frac{11}{2}\right)  _{2n}%
\ }=2^{4n}\frac{2^{-2n}\left(  4\right)  _{2n}2^{-2n}\left(  2\right)  _{2n}%
}{\ \left(  5\right)  _{2n}\left(  \frac{11}{2}\right)  _{2n}}\\
&  =\frac{\ \left(  4\right)  _{2n}\ \left(  2\right)  _{2n}}{\ \left(
5\right)  _{2n}\left(  \frac{11}{2}\right)  _{2n}}.
\end{align*}
Thus $X$ is (equidistributed as) the product of two independent random
variables $X_{1},X_{2}$ with%
\begin{align*}
\left\langle X_{1}^{n}\right\rangle  &  =\frac{\left(  4\right)  _{2n}%
}{\left(  5\right)  _{2n}}=\frac{4}{4+2n}=\frac{2}{2+n},\\
\left\langle X_{2}^{n}\right\rangle  &  =\frac{\left(  2\right)  _{2n}%
}{\left(  \frac{11}{2}\right)  _{2n}}.
\end{align*}
Clearly $X_{1}$ has the density $f_{1}\left(  t\right)  =2t,0\leq t\leq1$. The
density of $X_{2}$ is%
\[
f_{2}\left(  t\right)  =\frac{1}{2B\left(  2,\frac{7}{2}\right)  }\left(
1-\sqrt{t}\right)  ^{5/2},
\]
because
\begin{align*}
\int_{0}^{1}t^{n}f_{2}\left(  t\right)  dt  &  =\frac{1}{2B\left(  2,\frac
{7}{2}\right)  }\int_{0}^{1}t^{n}\left(  1-\sqrt{t}\right)  ^{5/2}dt\\
&  =\frac{1}{B\left(  2,\frac{7}{2}\right)  }\int_{0}^{1}s^{2n}s\left(
1-s\right)  ^{5/2}ds=\frac{\left(  2\right)  _{2n}}{\left(  \frac{11}%
{2}\right)  _{2n}}.
\end{align*}

The density $f\left(  x\right)  $ of $X$ is given by%
\begin{align*}
f\left(  t\right)   &  =\int_{x}^{1}f_{1}\left(  \frac{x}{s}\right)
f_{2}\left(  s\right)  \frac{ds}{s}\\
&  =\frac{2}{2B\left(  2,\frac{7}{2}\right)  }\int_{x}^{1}\frac{x}{s}\left(
1-\sqrt{s}\right)  ^{5/2}\frac{ds}{s}\\
&  =\frac{2t}{B\left(  2,\frac{7}{2}\right)  }\int_{\sqrt{x}}^{1}u^{-3}\left(
1-u\right)  ^{5/2}du\\
&  =\frac{63x}{2}\int_{\sqrt{x}}^{1}u^{-3}\left(  1-u\right)  ^{5/2}du.
\end{align*}

The integral is evaluated as follows: set $u=1-s^{2},du=-2sds$,%
\begin{align*}
f\left(  x\right)   &  =63x\int_{0}^{\sqrt{1-\sqrt{x}}}\frac{s^{6}}{\left(
1-s^{2}\right)  ^{3}}ds\\
&  =\frac{63x}{8}\left\{  \frac{-s\left(  15-25s^{2}+8s^{4}\right)  }{\left(
1-s^{2}\right)  ^{2}}+\frac{15}{2}\ln\frac{\left(  1+s\right)  ^{2}}{1-s^{2}%
}\right\}  _{s=0}^{s=\sqrt{1-\sqrt{x}}}\\
&  =\frac{63}{8}\left\{  \left(  1-\sqrt{x}\right)  ^{1/2}\left(  2-9\sqrt
{x}-8x\right)  +15x\ln\left(  1+\sqrt{1-\sqrt{x}}\right)  -\frac{15}{4}x\ln
x\right\}  .
\end{align*}

Also $f\left(  x\right)  =O\left(  \left(  1-x\right)  ^{7/2}\right)  $ near
$x=1$.

\subsection{Density of the determinant under the Bures metric} \label{DensityBures}

Using the Bures metric one obtains%
\[
\left\langle \left\vert \rho\right\vert ^{n}\right\rangle =\frac{\left(
\frac{1}{2}\right)  _{n}\left(  1\right)  _{n}\left(  \frac{3}{2}\right)
_{2n}}{2^{8n}\left(  \frac{3}{2}\right)  _{n}\left(  2\right)  _{n}\left(
4\right)  _{2n}}=\frac{2^{-8n}}{\left(  n+1\right)  \left(  2n+1\right)
}\frac{\left(  \frac{3}{2}\right)  _{2n}}{\left(  4\right)  _{2n}},
\]
for $n=0,1,2,\ldots$. As above we consider the random variable $X=2^{8}%
\left\vert \rho\right\vert $.

The density $f\left(  x\right)  $ of $X$, for $0<x\leq1$, satisfies
\[
\int_{0}^{1}x^{n}f\left(  x\right)  dx=\frac{1}{\left(  n+1\right)  \left(
2n+1\right)  }\frac{\left(  \frac{3}{2}\right)  _{2n}}{\left(  4\right)
_{2n}},n=0,1,2,\ldots.
\]
We express $X$ as the product of two random variables.

Let
\[
f_{1}\left(  t\right)  =t^{-1/2}-1,0<t\leq1,
\]
then
\[
\int_{0}^{1}t^{n}f_{1}\left(  t\right)  dt=\frac{1}{\left(  n+1\right)
\left(  2n+1\right)  },n=0,1,2,\ldots.
\]

Next observe (from equation \ref{betaX}):
\[
\frac{\Gamma\left(  4\right)  }{\Gamma\left(  \frac{3}{2}\right)
\Gamma\left(  \frac{5}{2}\right)  }\int_{0}^{1}s^{n}s^{1/2}\left(  1-s\right)
^{3/2}ds=\frac{\left(  \frac{3}{2}\right)  _{n}}{\left(  4\right)  _{n}},
\]
so set $s=t^{1/2}$ (and note $\frac{\Gamma\left(  4\right)  }{\Gamma\left(
\frac{3}{2}\right)  \Gamma\left(  \frac{5}{2}\right)  }=\frac{16}{\pi}$,
$ds=\frac{1}{2}t^{-1/2}dt$) to obtain%
\[
\frac{8}{\pi}\int_{0}^{1}t^{n}t^{-1/4}\left(  1-t^{1/2}\right)  ^{3/2}%
dt=\frac{\left(  \frac{3}{2}\right)  _{2n}}{\left(  4\right)  _{2n}%
},n=0,1,2\ldots.
\]
Let%
\[
f_{2}\left(  t\right)  =\frac{8}{\pi}t^{-1/4}\left(  1-t^{1/2}\right)
^{3/2},0<t\leq1.
\]
By Lemma \ref{X1X2} the desired density function is%
\begin{align*}
f\left(  x\right)   &  =\int_{x}^{1}f_{1}\left(  \frac{x}{t}\right)
f_{2}\left(  t\right)  \frac{dt}{t}\\
&  =\frac{8}{\pi}\int_{x}^{1}\left(  \left(  \frac{t}{x}\right)
^{1/2}-1\right)  t^{-1/4}\left(  1-t^{1/2}\right)  ^{3/2}dt\\
&  =\frac{8}{\pi\sqrt{x}}\int_{x}^{1}\left(  t^{1/2}-x^{1/2}\right)
t^{-1/4}\left(  1-t^{1/2}\right)  ^{3/2}dt.
\end{align*}
Substitute $t=s^{2}$, then
\[
f\left(  x\right)  =\frac{16}{\pi\sqrt{x}}\int_{\sqrt{x}}^{1}\left(
s-\sqrt{x}\right)  s^{1/2}\left(  1-s\right)  ^{3/2}ds,
\]
an elementary\ integral; indeed%
\[
f\left(  x\right)  =\frac{1}{\pi\sqrt{x}}\left\{  3\pi\left(  4\sqrt
{x}+1\right)  -4\left(  13+2\sqrt{x}\right)  \sqrt{\sqrt{x}-x}-2\left(
12\sqrt{x}+3\right)  \arcsin\left(  2\sqrt{x}-1\right)  \right\}  .
\]
As with the Hilbert Schmidt metric, $f\left(  x\right)  =O\left(  \left(
1-x\right)  ^{7/2}\right)  $ near $x=1$.

\subsection{The joint moments of $\left\vert \rho\right\vert $ and $\left\vert
\rho^{PT}\right\vert $}

The partial transpose $\rho^{PT}$ of $\rho$ is obtained by interchanging the
values of $\rho_{14}$ and $\rho_{23}$ (and $\rho_{41}$ and $\rho_{32}$). In this section
we introduce a conjecture for
\[
\left\langle \left\vert \rho^{PT}\right\vert ^{n}\left\vert \rho\right\vert
^{k}\right\rangle ,k,n=0,1,2,3,\ldots,
\]
using the density on $\Omega_{1}$ coming from the Hilbert-Schmidt metric.

For the upper triangular matrix $C$ and $\rho=C^{t}C$ we find%
\begin{align*}
&  \left\vert \rho^{PT}\right\vert \\
&  =c_{11}^{2}c_{22}^{2}c_{33}^{2}c_{44}^{2}+2c_{11}c_{22}(c_{11}c_{14}%
-c_{12}c_{13}-c_{22}c_{23})\\
&  \times(-c_{11}c_{23}c_{44}^{2}-c_{23}c_{34}^{2}c_{11}+c_{11}c_{33}%
c_{34}c_{24}+c_{22}c_{33}^{2}c_{14}\\
&  -c_{22}c_{13}c_{33}c_{34}-c_{12}c_{33}^{2}c_{24}+c_{23}c_{33}c_{34}%
c_{12})\\
&  -(c_{11}c_{14}-c_{12}c_{13}-c_{22}c_{23})^{2}\\
&  \times(4c_{22}c_{23}c_{11}c_{14}+c_{11}^{2}c_{44}^{2}+c_{11}^{2}c_{34}%
^{2}+c_{11}^{2}c_{24}^{2}-2c_{11}c_{13}c_{22}c_{24}-2c_{11}c_{12}c_{33}%
c_{34}\\
&  -2c_{11}c_{12}c_{23}c_{24}+c_{22}^{2}c_{13}^{2}-2c_{12}c_{13}c_{22}%
c_{23}+c_{22}^{2}c_{33}^{2}+c_{12}^{2}c_{33}^{2}+c_{12}^{2}c_{23}^{2})\\
&  -(c_{11}c_{14}-c_{12}c_{13}-c_{22}c_{23})^{4}.
\end{align*}
Of course $\left\vert \rho\right\vert =c_{11}^{2}c_{22}^{2}c_{33}^{2}%
c_{44}^{2}$. We introduce some utility functions (throughout $n,k=0,1,2,\ldots
$). For a rational function $F\left(  k\right)  =\frac{p\left(  k\right)
}{q\left(  k\right)  }$ of $k$ define the degree to be $\deg\left(  p\right)
-\deg\left(  q\right)  $.%
\begin{align*}
F_{0}\left(  k\right)   &  =\left\langle \left\vert \rho\right\vert
^{k}\right\rangle =\frac{\left(  1\right)  _{k}\left(  \frac{3}{2}\right)
_{k}\left(  2\right)  _{k}\left(  \frac{5}{2}\right)  _{k}}{\left(  10\right)
_{4k}},\\
F_{1}\left(  n,k\right)   &  =\left\langle \left\vert \rho^{PT}\right\vert
^{n}\left\vert \rho\right\vert ^{k}\right\rangle /\left\langle \left\vert
\rho\right\vert ^{k}\right\rangle ,\\
F_{2}\left(  n,k\right)   &  =\left\langle \left\vert \rho\right\vert
^{k}\left(  \left\vert \rho^{PT}\right\vert -\left\vert \rho\right\vert
\right)  ^{n}\right\rangle /\left\langle \left\vert \rho\right\vert
^{k}\right\rangle ,\\
R\left(  n,k\right)   &  =F_{0}\left(  n+k\right)  /F_{0}\left(  k\right)
=\frac{\left(  k+1\right)  _{n}\left(  k+\frac{3}{2}\right)  _{n}\left(
k+2\right)  _{n}\left(  k+\frac{5}{2}\right)  _{n}}{\left(  4k+10\right)
_{4n}}.
\end{align*}
(Note $F_{2}\left(  0,k\right)  =1=R\left(  0,k\right)  $). The goal is to
find (and prove) a closed form for $F_{1}\left(  n,k\right)  $, that is a
general formula. Direct computation for $n=1,2,3$ shows that $F_{1}\left(
n,k\right)  $ is rational in $k$ of degree $0$; and at first glance, does not
have an obvious formula (for the numerator). Some experimentation leads to the
observation that $F_{1}\left(  n,k\right)  -R\left(  n,k\right)  $ is of
degree $-2$ (verified only for small $n$). This motivates the investigation of
the decomposition%
\begin{align*}
\left\vert \rho^{PT}\right\vert ^{n} &  =\sum_{j=0}^{n}\binom{n}{j}\left\vert
\rho\right\vert ^{n-j}\left(  \left\vert \rho^{PT}\right\vert -\left\vert
\rho\right\vert \right)  ^{j}\\
F_{1}\left(  n,k\right)   &  =\sum_{j=0}^{n}\binom{n}{j}F_{2}\left(
j,k+n-j\right)  R\left(  n-j,k\right)  .
\end{align*}
For $n=1$ we compute%
\[
F_{1}\left(  1,k\right)  =R\left(  1,k\right)  -\frac{1}{16\left(
4k+13\right)  \left(  k+3\right)  };
\]
this is an encouraging result, and it implies $F_{2}\left(  1,k\right)
=-\frac{1}{16\left(  4k+13\right)  \left(  k+3\right)  }$ , of degree $-2$.
From the known value of $F_{1}\left(  2,k\right)  $ and the equation
\[
F_{1}\left(  2,k\right)  =R\left(  2,k\right)  +2F_{2}\left(  1,k+1\right)
R\left(  1,k\right)  +F_{2}\left(  2,k\right)
\]
we find
\[
F_{2}\left(  2,k\right)  =\frac{\left(  k+12\right)  \left(  2k+7\right)
}{256\left(  k+3\right)  \left(  k+4\right)  \left(  4k+11\right)  \left(
4k+13\right)  \left(  4k+17\right)  }.
\]
This is of degree $-3$, rather than the hoped-for $-4$, and the factor
$\left(  k+12\right)  $ is not of the \textquotedblleft good\textquotedblright%
\ type, a divisor of $\left(  k+1\right)  _{4}$. So we try to modify
$F_{2}\left(  2,k\right)  $ by adding a bit of $F_{2}\left(  1,k+1\right)
R\left(  1,k\right)  $; in fact%
\[
F_{2}\left(  2,k\right)  +\frac{2}{k+1}F_{2}\left(  1,k+1\right)  R\left(
1,k\right)  =\frac{3}{128\left(  k+3\right)  \left(  k+4\right)  \left(
4k+11\right)  \left(  4k+17\right)  },
\]
of degree $-4$. We now have a \textquotedblleft good\textquotedblright%
\ expansion of $F_{1}\left(  2,k\right)  $, namely%
\[
R\left(  2,k\right)  +\frac{2k}{k+1}F_{2}\left(  1,k+1\right)  R\left(
1,k\right)  +\left(  F_{2}\left(  2,k\right)  +\frac{2}{k+1}F_{2}\left(
1,k+1\right)  R\left(  1,k\right)  \right)  .
\]
The terms are of degree $0,-2,-4$ and each is an expression in linear factors.
Next we consider $F_{2}\left(  3,k\right)  $. This turns out to be of degree
$-5$ (rather than $-6$). Some effort leads to the satisfactory result:%
\begin{align*}
&  F_{2}\left(  3,k\right)  +\frac{6}{k+1}F_{2}\left(  2,k+1\right)  R\left(
1,k\right)  +\frac{12}{\left(  k+1\right)  \left(  k+2\right)  }F_{2}\left(
1,k+2\right)  R\left(  2,k\right)  \\
&  =-\frac{45}{2048\left(  k+3\right)  \left(  k+4\right)  \left(  k+5\right)
\left(  4k+11\right)  \left(  4k+13\right)  \left(  4k+21\right)  },\\
&  \frac{3\left(  k-1\right)  }{k+1}F_{2}\left(  2,k+1\right)  R\left(
1,k\right)  +\frac{6\left(  k-1\right)  }{\left(  k+1\right)  \left(
k+2\right)  }F_{2}\left(  1,k+2\right)  R\left(  2,k\right)  \\
&  =\frac{9\left(  k-1\right)  \left(  k+2\right)  \left(  2k+3\right)
}{4096\left(  k+3\right)  \left(  k+4\right)  \left(  k+5\right)  \left(
4k+11\right)  \left(  4k+13\right)  \left(  4k+15\right)  \left(
4k+21\right)  }.
\end{align*}
At this point there are enough examples to try to fit a formula to these
expansions. Indeed, for $0\leq j\leq n$ let
\begin{align*}
c_{j}\left(  n,k\right)   &  =\frac{1}{2^{6n}\left(  k+3\right)  _{n}\left(
2k+\frac{11}{2}\right)  _{2n}}\\
&  \times\frac{4^{j}n!\left(  \frac{1}{2}\right)  _{j}}{\left(  n-j\right)
!}\left(  -2k-2n-\frac{7}{2}\right)  _{j}\left(  k-j+1\right)  _{n-j}\left(
k+\frac{3}{2}\right)  _{n-j}\left(  k+2\right)  _{n-j},
\end{align*}
then
\begin{equation}
F_{1}\left(  n,k\right)  =\sum_{j=0}^{n}c_{j}\left(  n,k\right)
,\label{F1exp1}%
\end{equation}
is the conjectured formula. The degree of $c_{j}\left(  n,k\right)  $ is
$-2j$. We use the descending Pochhammer symbol $\left(  a\right)  _{\left(
n\right)  }=\prod_{i=1}^{n}\left(  a+1-i\right)  =\left(  -1\right)
^{n}\left(  -a\right)  _{n}$. If the conjecture is valid then for $1\leq j\leq
n$
\begin{gather*}
c_{j}\left(  n,k\right)  =\frac{\left(  n\right)  _{\left(  j\right)  }\left(
k\right)  _{\left(  j\right)  }}{j!\left(  n+k-1\right)  _{\left(
2j-1\right)  }}\\
\times\sum_{i=0}^{j-1}\frac{1}{i!}\left(  n+k-1\right)  _{\left(
j-1-i\right)  }\left(  i+j-1\right)  _{\left(  2i\right)  }F_{2}\left(
j-i,n+k+i-j\right)  R\left(  n+i-j,k\right)  .
\end{gather*}
This generalizes the examples found above. For generic $k$ there is the
expression
\begin{align}
F_{1}\left(  n,k\right)   &  =\frac{\left(  k+1\right)  _{n}\left(  k+\frac
{3}{2}\right)  _{n}\left(  k+2\right)  _{n}}{2^{6n}\left(  k+3\right)
_{n}\left(  2k+\frac{11}{2}\right)  _{2n}}\label{F1exp2}\\
&  \times~_{5}F_{4}\left(
\genfrac{}{}{0pt}{}{-n,1,\frac{1}{2},-k,-2k-2n-\frac{7}{2}}{-k-n-1,-k-n-\frac
{1}{2},-\frac{k+n}{2},-\frac{k+n-1}{2}}%
;1\right)  .\nonumber
\end{align}
This sum is a terminating balanced hypergeometric series. (\textquotedblleft
balanced\textquotedblright\ means the sum of the numerator parameters + 1
equals the sum of the denominator parameters.) However the $_{5}F_{4}$-sum is
symmetric in $\left(  n,k\right)  $ and the summation range is $0\leq
j\leq\min\left(  n,k\right)  $. When $0\leq k<n$ this omits the terms in the
first formula for the range $0\leq n-j<\frac{n-k}{2}$. For this case the best
way is to use equation (\ref{F1exp1}) (or else use (\ref{F1exp2}) with generic
$k$ to compute the rational function, then substitute the desired integer
value for $k$).

The special case $k=0$ is:%
\[
F_{1}\left(  n,0\right)  =\frac{2\left(  2n+1\right)  !}{2^{8n}\left(
n+2\right)  \left(  \frac{11}{2}\right)  _{2n}}+\frac{\left(  2n\right)
!\left(  -2n-\frac{7}{2}\right)  _{n}}{2^{6n}\left(  3\right)  _{n}\left(
\frac{11}{2}\right)  _{2n}}~_{4}F_{3}\left(
\genfrac{}{}{0pt}{}{-\frac{n-2}{2},-\frac{n-1}{2},2,\frac{3}{2}}{\frac{1}%
{2}-n,1-n,\frac{9}{2}+n}%
;1\right)  .
\]
Another interesting special case is $k=n$:%
\begin{gather*}
\left\langle \left\vert \rho^{PT}\right\vert ^{n}\left\vert \rho\right\vert
^{n}\right\rangle =F_{1}\left(  n,n\right)  F_{0}\left(  n\right) \\
=\frac{\left(  2n\right)  !\left(  \frac{3}{2}\right)  _{2n}}{2^{12n}\left(
\frac{11}{2}\right)  _{4n}\left(  n+1\right)  }~_{4}F_{3}\left(
\genfrac{}{}{0pt}{}{\ -n,1,\frac{1}{2},-4n-\frac{7}{2}}{-2n-1,-2n-\frac{1}%
{2},\frac{1}{2}-n}%
;1\right)  .
\end{gather*}
The conjecture for $F_{1}\left(  n,k\right)  $ has been checked by
computer-aided symbolic algebra up to $n=13$.

\subsection{Gaussian quadrature} \label{Gaussian}

The method of Gaussian quadrature based on orthogonal polynomials can be
applied to the density problem (see \cite[Thm. 3.4.2, p.48]{szego}). Suppose
$\mu$ is a probability measure supported on a bounded interval $\left[
a,b\right]  $ and the moments are $\mu_{j}:=\int_{a}^{b}x^{j}d\mu\left(
x\right)  $. The orthogonal polynomials $\left\{  P_{n}\left(  x\right)
:n=0,1,2,\ldots\right\}  $ for $\mu$ (where $P_{n}$ is of degree $n$ and
$\int_{a}^{b}x^{j}P_{n}\left(  x\right)  d\mu\left(  x\right)  =0$ for $0\leq
j<n$) are determined by the moment sequence. Solve the linear system%
\[
\sum_{i=0}^{n-1}a_{i}\mu_{i+j}=-\mu_{j+n},0\leq j\leq n-1
\]
to obtain the coefficients $\left\{  a_{i}\right\}  $ for the monic orthogonal
polynomial%
\[
P_{n}\left(  x\right)  =x^{n}+\sum_{i=0}^{n-1}a_{i}x^{i}.
\]
Then $P_{n}$ has $n$ distinct zeros $\lambda_{1}<\lambda_{2}<\ldots
<\lambda_{n}$, contained in $\left(  a,b\right)  .$The structural constant
$h_{n}=\int_{a}^{b}P_{n}\left(  x\right)  ^{2}d\mu\left(  x\right)
=\sum_{i=0}^{n-1}a_{i}\mu_{i+n}+\mu_{2n}$. The Gaussian quadrature rule with
$n$ nodes is%
\begin{align*}
\mathcal{G}_{n}\left(  p\right)   &  =\sum_{i=1}^{n}w_{n,i}p\left(
\lambda_{i}\right)  ,\\
w_{n,i}  &  =\frac{h_{n-1}}{P_{n}^{^{\prime}}\left(  \lambda_{i}\right)
P_{n-1}\left(  \lambda_{i}\right)  },1\leq i\leq n.
\end{align*}
Then $\mathcal{G}_{n}\left(  p\right)  =\int_{a}^{b}pd\mu$ for all polynomials
$p$ of degree $\leq2n-1$. The sequence of discrete measures $\left\{
\mathcal{G}_{n}:n=2,3,\ldots\right\}  $ converges weak-* (in the dual space of
$C\left[  a,b\right]  $) to the measure $\mu$. The piecewise linear graph
formed by consecutively joining $\left[  a,0\right]  $, $\left[  \frac{1}%
{2}\left(  \lambda_{1}+\lambda_{2}\right)  ,w_{1}\right]  $, $\left[  \frac
{1}{2}\left(  \lambda_{2}+\lambda_{3}\right)  ,w_{1}+w_{2}\right]  $,
$\ldots\left[  \frac{1}{2}\left(  \lambda_{i}+\lambda_{i+1}\right)
,\sum_{j=1}^{i}w_{j}\right]  $, $\ldots,\left[  b,1\right]  $ is an
approximation to the cumulative distribution function of $\mu$. (Consider this
as a sort of mid-point integration rule.)

The orthonormal polynomials satisfy the three-term recurrence%
\[
xp_{n}\left(  x\right)  =\alpha_{n}p_{n+1}\left(  x\right)  +\beta_{n}%
p_{n}\left(  x\right)  +\alpha_{n-1}p_{n-1}\left(  x\right)  ,p_{-1}%
=0,p_{1}=1
\]
The most common approach to the computations is to find the coefficients
$\left\{  \alpha_{i},\beta_{i}\right\}  $ directly from the moments. This is
known to be a numerically ill-conditioned problem, so a relatively large
number of significant digits must be used in the calculation. The algorithm of
\cite[p.476]{Sack} with 30-digit floating-point arithmetic was used here. The
computed values were checked for accuracy by evaluating the errors
\[
\varepsilon_{j}=\mu_{j}-\sum_{i=1}^{n}w_{n,i}\lambda_{i}^{j},0\leq j\leq2n-1.
\]
For the moments of $16\left\vert \rho^{PT}\right\vert $ and $n=20$ we obtain

$%
\begin{bmatrix}
\lambda & -.9501 & -.9081 & -.8587 & -.8024 & -.7402\\
w & .2714\cdot10^{-10} & .1397\cdot10^{-8} & .2416\cdot10^{-7} &
.2337\cdot10^{-6} & .1553\cdot10^{-5}%
\end{bmatrix}
\medskip$

$%
\begin{bmatrix}
\lambda & -.6734 & -.6032 & -.5309 & -.4581 & -.3860\\
w & .7908\cdot10^{-5} & .3293\cdot10^{-4} & .1171\cdot10^{-3} & .3669\cdot
10^{-3} & .1034\cdot10^{-2}%
\end{bmatrix}
\medskip$

$%
\begin{bmatrix}
\lambda & -.3160 & -.2495 & -.1877 & -.1317 & -.08248\\
w & .2671\cdot10^{-2} & .6408\cdot10^{-2} & .1446\cdot10^{-1} & .3111\cdot
10^{-1}\  & .6499\cdot10^{-1}%
\end{bmatrix}
\medskip$

$%
\begin{bmatrix}
\lambda & -.04104 & -.008293\  & .01040 & .02973 & .04698\\
w & .1372 & .3467 & .3440 & .4894\cdot10^{-1} & .1994\cdot10^{-2}%
\end{bmatrix}
.\medskip$

One observes that a majority of the zeros are in $\left[  -1,0\right]  $ and
most of the mass is contained in $\left[  -0.05,0.05\right]  $. With $n=30$ we
find $4$ zeros in $\left(  0,\frac{1}{16}\right)  $; linear interpolation of
the c.d.f. gives $\Pr\left\{  \left\vert \rho^{PT}\right\vert >0\right\}
\simeq0.42924$.

The distribution of $2^{16}\left\vert \rho\right\vert \left\vert \rho
^{PT}\right\vert $ is somewhat more spread out over the interval. For $n=30$
we find $17$ zeros in $\left(  0,1\right)  $, and linear interpolation yields
$\Pr\left\{  \left\vert \rho^{PT}\right\vert >0\right\}  \simeq0.46129$.

\subsection{Conjectures for the complex case} \label{ComplexConjectures}

Here we consider the question of moments of $\left\vert \varrho^{PT}%
\right\vert $ when $\rho$ is a 4-by-4 Hermitian positive-definite matrix of
trace one. The conjectured formulae have been verified for $n=1,2,3,4$ (see
the previous sections). The conjecture was arrived at by inspecting the real
case and using an analogous approach to the computed examples. It is
interesting that the real and complex conjectured formulae can be combined
into one formula with a parameter $\alpha$. Set $\alpha=\frac{1}{2}$ for the
real case, $\alpha=1$ for the complex case. One could speculate whether
$\alpha=2$ is related to a quaternionic or symplectic version \cite{andai,dorje,quartic}. The general
formulae are%
\begin{gather*}
\left\langle \left\vert \rho\right\vert ^{k}\right\rangle =\frac{k!\left(
\alpha+1\right)  _{k}\left(  2\alpha+1\right)  _{k}}{2^{6k}\left(
3\alpha+\frac{3}{2}\right)  _{k}\left(  6\alpha+\frac{5}{2}\right)  _{2k}},\\
\left\langle \left\vert \rho^{PT}\right\vert ^{n}\left\vert \rho\right\vert
^{k}\right\rangle /\left\langle \left\vert \rho\right\vert ^{k}\right\rangle
=\frac{1}{2^{6n}\left(  k+3\alpha+\frac{3}{2}\right)  _{n}\left(
2k+6\alpha+\frac{5}{2}\right)  _{2n}}\\
\times\sum_{j=0}^{n}4^{j}\binom{n}{j}\left(  \alpha\right)  _{j}\left(
\alpha+\frac{1}{2}\right)  _{j}\left(  k-j+1\right)  _{n-j}\\
\times\left(  -2k-2n-1-5\alpha\right)  _{j}\left(  k+1+\alpha\right)
_{n-j}\left(  k+2+\alpha\right)  _{n-j}.
\end{gather*}
For generic $k$ this formula can be written as%
\begin{align*}
&  \left\langle \left\vert \rho^{PT}\right\vert ^{n}\left\vert \rho\right\vert
^{k}\right\rangle /\left\langle \left\vert \rho\right\vert ^{k}\right\rangle
\\
&  =\frac{\left(  k+1\right)  _{n}\left(  k+1+\alpha\right)  _{n}\left(
k+1+2\alpha\right)  _{n}}{2^{6n}\left(  k+3\alpha+\frac{3}{2}\right)
_{n}\left(  2k+6\alpha+\frac{5}{2}\right)  _{2n}}\\
&  \times~_{5}F_{4}\left(
\genfrac{}{}{0pt}{}{-n,-k,\alpha,\alpha+\frac{1}{2},-2k-2n-1-5\alpha
}{-k-n-\alpha,-k-n-2\alpha,-\frac{k+n}{2},-\frac{k+n-1}{2}}%
;1\right)  .
\end{align*}
The special case $n=k$ is%
\begin{gather*} \label{nequalk}
\left\langle \left\vert \rho\right\vert ^{n}\left\vert \rho^{PT}\right\vert
^{n}\right\rangle \\
=\frac{\left(  2n\right)  !\left(  1+\alpha\right)  _{2n}\left(
1+2\alpha\right)  _{2n}}{2^{12n}\left(  3\alpha+\frac{3}{2}\right)
_{2n}\left(  6\alpha+\frac{5}{2}\right)  _{4n}}~_{4}F_{3}\left(
\genfrac{}{}{0pt}{}{\ -n,\alpha,\alpha+\frac{1}{2},-4n-1-5\alpha
}{-2n-\alpha,-2n-2\alpha,\frac{1}{2}-n}%
;1\right)  .
\end{gather*}
For $k=0$ we have%
\begin{gather*} \label{nequalzero}
\left\langle \left\vert \rho^{PT}\right\vert ^{n}\right\rangle =\frac
{n!\left(  \alpha+1\right)  _{n}\left(  2\alpha+1\right)  _{n}}{2^{6n}\left(
3\alpha+\frac{3}{2}\right)  _{n}\left(  6\alpha+\frac{5}{2}\right)  _{2n}}\\
+\frac{\left(  -2n-1-5\alpha\right)  _{n}\left(  \alpha\right)  _{n}\left(
\alpha+\frac{1}{2}\right)  _{n}}{2^{4n}\left(  3\alpha+\frac{3}{2}\right)
_{n}\left(  6\alpha+\frac{5}{2}\right)  _{2n}}~_{5}F_{4}\left(
\genfrac{}{}{0pt}{}{-\frac{n-2}{2},-\frac{n-1}{2},-n,\alpha+1,2\alpha
+1}{1-n,n+2+5\alpha,1-n-\alpha,\frac{1}{2}-n-\alpha}%
;1\right)  ;
\end{gather*}
because of the denominator parameter $1-n$ it is necessary to replace the
$_{5}F_{4}$-sum by $1$ to obtain the correct value when $n=1$.
\subsection{Lower-dimensional ("non-generic") case study} \label{fourparameter}
In the Cholesky method, set five of the off-diagonal entries to zero; the positive
matrix $\rho$ is $%
\begin{bmatrix}
x_{1}^{2} & 0 & 0 & 0\\
0 & x_{2}^{2} & x_{2}x_{5} & 0\\
0 & x_{2}x_{5}^{\ast} & x_{3}^{2}+x_{5}x_{5}^{\ast} & 0\\
0 & 0 & 0 & x_{4}^{2}%
\end{bmatrix}
$ 
(where $x_{i}\geq0,1\leq i\leq4$, and $x_{5}$ comes from $\mathbb{R}^{\beta
}$, equipped with an algebra structure including conjugation and a norm, e.g. $\beta=2,\mathbb{C}$ )

Then $|\rho|=x_{1}^{2}x_{2}^{2}x_{3}^{2}x_{4}^{2}$, and $|\rho^{PT}|=x_{2}%
^{2}\left(  x_{3}^{2}+x_{5}x_{5}^{\ast}\right)  \left(  x_{1}^{2}x_{4}%
^{2}-x_{2}^{2}x_{5}x_{5}^{\ast}\right)  .$ Consider $\rho$ as an element of
$\mathbb{R}^{4+\beta}$; then the Jacobian for the map $C\mapsto C^{\ast}C$ is
$J=16x_{1}x_{2}^{\beta+1}x_{3}x_{4}$. Write $x_{5}x_{5}^{\ast}=\left\vert
x_{5}\right\vert ^{2}$. With the usual Dirichlet integral techniques,
integrating over the unit sphere $\sum_{i=1}^{4}x_{i}^{2}+\left\vert
x_{5}\right\vert ^{2}=1$ we get the normalized integral%
\begin{align*}
& \int x_{1}^{2m_{1}}x_{2}^{2m_{2}}x_{3}^{2m_{3}}x_{4}^{2m_{4}}\left\vert
x_{5}\right\vert ^{2m_{5}}\left(  |\rho|\right)  ^{k}J\left(  x\right)
d\mu\\
& =\frac{\left(  k+1\right)  _{m_{1}}\left(  k+1\right)  _{m_{3}}\left(
k+1\right)  _{m_{4}}\left(  k+1+\frac{\beta}{2}\right)  _{m_{2}}\left(
\frac{\beta}{2}\right)  _{m_{5}}}{\left(  4+\beta+4k\right)  _{\left\vert
m\right\vert }}\delta\left(  k\right)  ,
\end{align*}
where $\left\vert m\right\vert =\sum_{i=1}^{5}m_{i}$; and%
\[
\delta\left(  k\right)  :=\int\left(  |\rho|\right)  ^{k}J\left(  x\right)
d\mu=\frac{k!^{3}\left(  1+\frac{\beta}{2}\right)  _{k}}{\left(
4+\beta\right)  _{4k}}.
\]
Then%
\begin{align*}
& \int\left(  |\rho^{PT}|\right)  ^{n}\left(  |\rho|\right)  ^{k}J\left(
x\right)  d\mu\\
& =\delta\left(  k\right)  \frac{\left(  k+1\right)  _{n}^{2}\left(
k+1+\frac{\beta}{2}\right)  _{n}^{2}}{\left(  4+\beta+4k\right)  _{4n}}\\
& \times~_{4}F_{3}\left(
\genfrac{}{}{0pt}{}{-n,k+1+\frac{\beta}{2}+n,k+1+\frac{\beta}{2}+n,\frac
{\beta}{2}}{-k-n,-k-n,k+1+\frac{\beta}{2}}%
;1\right)  .
\end{align*}
Proof: Expanding%
\begin{align*}
\left(  |\rho^{PT}|\right)  ^{n}  & =x_{2}^{2n}\left(  x_{3}^{2}+\left\vert
x_{5}\right\vert ^{2}\right)  ^{n}\left(  x_{1}^{2}x_{4}^{2}-x_{2}%
^{2}\left\vert x_{5}\right\vert ^{2}\right)  ^{n}\\
& =\sum_{i,j=0}^{n}\binom{n}{i}\binom{n}{j}\left(  -1\right)  ^{j}%
x_{1}^{2n-2j}x_{2}^{2n+2j}x_{3}^{2n-2i}x_{4}^{2n-2j}\left\vert x_{5}%
\right\vert ^{2i+2j}.
\end{align*}
Now integrate with the above formula (value divided by $\delta\left(
k\right)  $) to obtain%
\begin{align*}
& \frac{1}{\left(  4+\beta+4k\right)  _{4n}}\sum_{i,j=0}^{n}\binom{n}{i}%
\binom{n}{j}\left(  -1\right)  ^{j}\left(  k+1\right)  _{n-j}^{2}\left(
k+1+\frac{\beta}{2}\right)  _{n+j}\\
& \times\left(  k+1\right)  _{n-i}\left(  \frac{\beta}{2}\right)  _{i+j}\\
& =\frac{1}{\left(  4+\beta+4k\right)  _{4n}}\sum_{j=0}^{n}\binom{n}{j}\left(
-1\right)  ^{j}\left(  k+1\right)  _{n-j}^{2}\left(  k+1+\frac{\beta}%
{2}\right)  _{n+j}\left(  \frac{\beta}{2}\right)  _{j}\\
& \times\sum_{i=0}^{n}\binom{n}{i}\left(  k+1\right)  _{n-i}\left(
\frac{\beta}{2}+j\right)  _{i},
\end{align*}
by the Chu-Vandermonde sum, the second line equals $\left(  k+1+\frac{\beta}%
{2}+j\right)  _{n}$. Use the substitutions
\begin{align*}
\left(  k+1\right)  _{n-j}  & =\left(  -1\right)  ^{j}\dfrac{\left(
k+1\right)  _{n}}{\left(  -k-n\right)  _{j}},\\
\left(  k+1+\frac{\beta}{2}+j\right)  _{n}  & =\frac{\left(  k+1+\frac{\beta
}{2}\right)  _{n}\left(  k+1+\frac{\beta}{2}+n\right)  _{j}}{\left(
k+1+\frac{\beta}{2}\right)  _{j}},\\
\binom{n}{j}\left(  -1\right)  ^{j}  & =\frac{\left(  -n\right)  _{j}}{j!}%
\end{align*}
in the $j$-sum to produce the stated formula.

Example:
\begin{gather*}
\int\left(  |\rho^{PT}|\right)  \left(  |\rho|\right)  ^{k}J\left(
x\right)  d\mu=\frac{\delta\left(  k\right)  }{\left(  4+\beta+4k\right)
_{4}}\\
\times\frac{1}{4}\left(  2k+2+\beta\right)  \left\{  \left(  k+1\right)
^{2}\left(  2k+2+\beta\right)  -\frac{1}{4}\beta\left(  2k+4+\beta\right)
^{2}\right\}
\end{gather*}

\begin{acknowledgments}
I would like to express appreciation to the Kavli Institute for Theoretical
Physics (KITP)
for computational support in this research, and Christian Krattenthaler, Mihai Putinar, Robert Mnatsakanov, Mark Coffey, Karol \.{Z}yczkowski and \"{O}mer E\u{g}ecio\u{g}lu for various communications. Further, Serge Provost, Jean Lasserre, Partha Biswas and Luis G. Medeiros de Souza provided guidance on reconstruction of probability distributions from moments. The earlier stages of the computations were greatly assisted by the Mathematica expertise of Michael Trott, and the later stages by the mathematical insights and suggestions of  Charles Dunkl. A referee requested greater clarification of certain basic concepts.
\end{acknowledgments}

\bibliography{WW15C}

\end{document}